\documentclass[journal,draftcls,onecolumn,11pt,twoside]{IEEEtranTCOM}
\ifCLASSINFOpdf

\else

\fi

\usepackage{mathrsfs}
\usepackage{amsmath}
\usepackage{amsfonts}
\usepackage{amsthm}
\usepackage{amssymb}
\usepackage{graphicx}
\usepackage{subfigure}
\usepackage{indentfirst}
\usepackage{array}
\usepackage{cite}
\usepackage{enumerate}
\usepackage{bm}
\usepackage[linesnumbered,ruled,vlined]{algorithm2e}
\usepackage{algorithmic}
\usepackage{multirow}
\usepackage{epstopdf}
\usepackage{verbatim}
\usepackage{stfloats}
\usepackage{color}
\usepackage{caption}
\usepackage{booktabs}
\usepackage{footnote}

\newtheorem{theorem}{\bf{Theorem}}
\newtheorem{lemma}{\bf{Lemma}}
\newtheorem{proposition}{\bf{Proposition}}

\newtheorem{conjecture}{\bf{Conjecture}}

\newtheorem{remark}{\bf{Remark}}
\SetKwComment{Comment}{//}{}

\begin{document}

\title{Performance Analysis for Polar Codes under Successive Cancellation List Decoding with Fixed List Size}
\author{Jinnan Piao, Dong Li, Xueting Yu, Zhibo Li, Ming Yang, Jindi Liu, and Peng Zeng
\thanks{
This work is supported in part by the National Key R\&D Program of China under Grant 2022YFB3207400, in part by the National Natural Science Foundation of China under Grant 62201562 and in part by the Fundamental Research Project of SIA under Grant 2022JC1K08. (\emph{Corresponding Author: Dong Li})
}
\thanks{
The authors are with the State Key Laboratory of Robotics, Shenyang Institute of Automation, Chinese Academy of Sciences, Shenyang 110016, China, with the Key Laboratory of Networked Control Systems,  Chinese Academy of Sciences, Shenyang 110016, China, and also with the Institutes for Robotics and Intelligent Manufacturing, Chinese Academy of Sciences, Shenyang 110169, China.
(e-mail: piaojinnan@sia.cn; lidong@sia.cn; yuxueting@sia.cn; lizhibo@sia.cn; yangming@sia.cn; liujindi@sia.cn; zp@sia.cn).
}
}

\maketitle
\begin{abstract}

In this paper, we first indicate that the block error event of polar codes under successive cancellation list (SCL) decoding is composed of path loss (PL) error event and path selection (PS) error event, where the PL error event is that correct codeword is lost during the SCL decoding and the PS error event is that correct codeword is reserved in the decoded list but not selected as the decoded codeword. Then, we simplify the PL error event by assuming the all-zero codeword is transmitted and derive the probability lower bound via the joint probability density of the log-likelihood ratios of information bits. Meanwhile, the union bound calculated by the minimum weight distribution is used to evaluate the probability of the PS error event. With the performance analysis, we design a greedy bit-swapping (BS) algorithm to construct polar codes by gradually swapping information bit and frozen bit to reduce the performance lower bound of SCL decoding. The simulation results show that the BLER performance of SCL decoding is close to the lower bound in the medium to high signal-to-noise ratio region and we can optimize the lower bound to improve the BLER performance of SCL decoding by the BS algorithm.

\end{abstract}


\IEEEpeerreviewmaketitle

\section{Introduction}

\subsection{Relative Research}

\IEEEPARstart{P}{olar} codes, invented by Ar{\i}kan, have been proved to achieve the capacity of arbitrary binary-input discrete memoryless channels (B-DMCs) as the code length goes to infinity with the successive cancellation (SC) decoding \cite{arikan}.
However, the performance of polar codes is weak at short to moderate code lengths under SC decoding. To improve the block error rate (BLER) performance, the successive cancellation list (SCL) decoding \cite{niuscl,talvardyscl} is proposed.
Since polar codes demonstrate advantages in error performance and other attractive application prospects, they have been adopted as the coding scheme for the control channel of the enhanced Mobile Broadband (eMBB) service category in the fifth generation wireless communication systems (5G) \cite{3GPP_5G_polar}.

To analyze the BLER performance of polar codes, the performance upper bound of SC decoding is proposed in \cite{arikan}, which is the sum of the error probabilities of all the information bits.
The Bhattacharyya parameter \cite{arikan} is first used to calculate the error probability of information bit for the binary erasure channel (BEC).
For other B-DMCs, the density evolution (DE) algorithm, initially proposed in \cite{DE} and improved in \cite{TalVardy}, tracks the probability distribution of the log-likelihood ratio (LLR) of each synthetic channel and accurately calculates the error probability of each information bit with a high computational cost.
For the binary-input additive white Gaussian noise (BI-AWGN) channel, the Gaussian approximation (GA) algorithm \cite{GA,GA_DAI} approximates the probability distribution of the LLR as Gaussian distribution and gives accurate error probability evaluation with limited complexity.
For fading channels, the average mutual information (AMI) equivalence \cite{ConstructionRayleighFading} finds an equivalent BI-AWGN channel of the given Rayleigh fading channel with the identical AMI to apply GA algorithm to calculate the error probability of information bit.
The above SC upper bound coincides with the performance of polar codes under SC decoding in the medium to high signal-to-noise ratio (SNR) region.
However, there is a large gap between the SC upper bound and the performance of SCL decoding.

For the performance analysis of SCL decoding, the theoretical question: ``How large should the list be to achieve maximum likelihood (ML) decoding'' is motivate by the property that the SCL decoding with limited list size can achieve nearly the ML performance.
To solve the question, Hashemi \emph{et al.} \cite{ReqiredListBEC} first prove that the required list size for an $\left(N, K\right)$ polar code in BEC is $L = 2^{K-\tau}$, where $\tau$ is the number of information bits that appear after the last frozen bit.
Then, Fazeli \emph{et al.} \cite{ReqiredListBDMC} extend the theory and prove that the required list size in B-DMCs is $L = 2^{\gamma}$, where $\gamma$ is the mixing factor.
Meanwhile, Co\c{s}kun \emph{et al.} \cite{InformationSCL} provide the upper and lower bound of the average required list size achieving the ML decoding by the entropies of the synthetic channels transmitting information bits.
However, these works are unrelated to the performance evaluation of SCL decoding with fixed list size.

\subsection{Motivation}

Since the SCL decoding with fixed list size is widely used in practice, there is another theoretical question: ``How to evaluate the performance of SCL decoding with fixed list size''. The SC upper bound can be regarded as the performance upper bound of SCL decoding with $L = 1$. Nevertheless, theory currently falls short of evaluating the performance of SCL decoding with $L \ge 2$. Thus, we try to extend the analysis process of SC upper bound to match the decoded process of SCL decoding and analyse the performance.

\subsection{Main Contributions}

In this paper, we analyse the BLER performance of polar codes under SCL decoding with fixed list size. We first indicate the block error event of SCL decoding is composed of path loss (PL) error event and path selection (PS) error event. Then, we derive the probability lower bound of PL error event and use the union bound of ML performance to evaluate the probability of PS error event. Finally, with the performance analysis, we design a greedy bit-swapping (BS) algorithm to improve the BLER performance of polar codes under SCL decoding.

The main contributions of this paper are summarized as follows.

\begin{enumerate}
  \item We first indicate the block error event of polar codes with SCL decoding is composed of PL error event and PS error event, where PL error event is that the correct codeword is lost during the SCL decoding and PS error event is that the correct codeword is reserved in the decoded list after SCL decoding but not selected as the decoded codeword.
      The BLER performance of SCL decoding is the sum of the probabilities of the two error events.
      Then, We simplify the PL error event by assuming the all-zero codeword is transmitted and use the union bound of the ML performance calculated by the minimum weight distribution (MWD) to evaluate the probability of the PS error event.
  \item We derive the probability lower bound of the simplified PL error event via the joint probability density of the LLRs of information bits. The simplified PL error event is first divided into several subevents by the judgement ranges of the LLRs of information bits.
      Then, the lower bound of the probability of each subevent is derived by the joint probability density of the LLRs.
      With the above process, we provide a modified upper bound of the SC decoding, a lower bound of the SCL decoding with $L = 2$, an approximate lower bound of the SCL decoding with $L = 4$ and an approximate performance of the SCL decoding with $L \ge 8$.
  \item We design a greedy BS algorithm to reduce the performance lower bound of SCL decoding by gradually swapping information bit and frozen bit in order to construct polar codes.
      In the construction algorithm, we first initialize the information set by the MWD sequence which is the optimum construction sequence under ML decoding. Then, we decide a bit-searching range by the Hamming weight and swap the information bit and the frozen bit in the range to reduce the lower bound of the SCL decoding.
      When swapping bit can not reduce the lower bound, we enlarge the range and continue swapping bits to reduce the lower bound until the range can not be enlarged.
\end{enumerate}

The simulation results first show that the BLER performance of SCL decoding approaches the lower bound in the medium to high SNR region.
Interestingly, for the polar codes constructed by GA algorithm, the probability of the PL error event is much less than the ML performance, which provides
an explanation for the phenomenon that the BLER performance of polar codes with limited list size can approach the ML performance.
Then, the simulation results also indicate that the lower bound can be
used to evaluate the BLER performance of CRC-polar concatenated codes under SCL decoding.
Finally, the simulation results
illustrate that we can optimize the proposed lower bound to improve the BLER performance of SCL decoding by the BS algorithm.

The remainder of the paper is organized as follows. Section
II describes the preliminaries of polar codes, SCL decoding, upper bound of SC decoding and simplified code tree of polar codes.
In Section III, we provide an outline of the performance analysis of SCL decoding and indicate the composition of the corresponding block error event.
The probability lower bound of the PL error event is derived in Section IV.
Section V describes the proposed greedy BS algorithm.
Section VI shows the lower bound of the BLER performance of SCL decoding and the performance of polar codes constructed by the BS algorithm. Section VII concludes this paper.

\section{Notations and Preliminaries}

\subsection{Notation Conventions}

In this paper, the lowercase letters, e.g., $x$, are used to denote scalars. The bold lowercase letters (e.g., ${\bf{x}}$) are used to denote vectors.
Notation ${{\bf x}_i^j}$ denotes the vector $(x_i,\cdots,x_j)$ and $x_i$ denotes the $i$-th element of ${\bf{x}}$.
The random variables are denoted by sans serif characters e.g., $\mathsf X$.
The sets are denoted by calligraphic characters, e.g., $\cal{X}$, and the notation $|\cal{X}|$ denotes the cardinality of $\cal{X}$.
The bold capital letters, such as $\bf{X}$, are used to denote matrices.
Throughout this paper, we write ${\bf{F}}^{\otimes n}$ to denote the $n$-th Kronecker power of $\bf{F}$.
$\bf 0$ means an all-zero vector.
$\ln \left(  \cdot  \right)$ and $\log \left(  \cdot  \right)$ mean ``base $e$ logarithm'' and ``base 2 logarithm'', respectively.
$y = \max^i_{x \in {\mathcal X}} f(x)$ means that $y$ is the $i$-th largest $f(x)$ for $x \in {\mathcal X}$ and $[\![N]\!]$ represents the set $\left\{1,2,\cdots,N\right\}$.

\subsection{Polar Codes}

For a B-DMC $W:{\mathcal X}\rightarrow{\mathcal Y}$ with input alphabet ${\mathcal X} = \left\{0,1\right\}$, output alphabet ${\mathcal Y}$ and transition probabilities $W\left(y|x\right)$, an $\left(N, K\right)$ polar code transforms $N$ identical copies of $W$ into synthetic channels $W_N^{\left(i\right)}: {\mathcal X}\rightarrow{\mathcal Y}^N \times {\mathcal X}^{i-1}$, $i\in [\![N]\!]$ with the transition probabilities
$W_N^{\left(i\right)}\left(\left.{\bf y}_1^N, {\bf u}_1^{i-1}\right|u_i\right)$.
The transformation depends on the polarization effect of the matrix
${{\bf F} = \left[
\begin{smallmatrix}
1&0\\
1&1
\end{smallmatrix}
\right]}$ and the codeword ${\bf x}_1^N$ of the polar code is calculated by the linear transform ${\bf x}_1^N  = {\bf u}_1^N{\bf G}$, where ${\bf G} = {\bf F}^{\otimes n}$ is the generator matrix and $n = \log N$.
The $N$-length information sequence ${\bf u}_1^N$ is  generated by assigning $u_i \in \left\{0,1\right\}$ to information bit if $i \in {\mathcal A}$ and $u_i$ to frozen bit if $i \in {\mathcal F}$,
where ${\mathcal A}$ is the information set with cardinality $|{\mathcal A}|=K$ including the indices of the $K$ most reliable synthetic channels and ${\mathcal F} = [\![N]\!] \setminus {\mathcal A}$ is the complementary set of ${\mathcal A}$.

Given the channel output vector ${\bf y}_1^N$ and the estimation ${\bf{\widehat u}}_1^{i - 1}$ of ${\bf u}_1^{i-1}$, the SC decoding with the decoding complexity $O(N\log N)$ calculates the LLR
\begin{equation}\label{EqLLR}
\theta_N^{\left(i\right)}\left({\bf y}_1^N, {\bf {\widehat u}}_1^{i-1}\right)
= \ln\frac{W_N^{\left(i\right)}\left(\left.{\bf y}_1^N, {\bf {\widehat u}}_1^{i-1}\right|u_i = 0\right)}{W_N^{\left(i\right)}\left(\left.{\bf y}_1^N, {\bf {\widehat u}}_1^{i-1}\right|u_i = 1\right)}
\end{equation}
and estimates $u_i$ by
\begin{equation}
{\widehat u_i} = \left\{ \begin{array}{l}
0,~~~~~{\rm if}~\theta_N^{\left(i\right)}\left({\bf y}_1^N, {\bf {\widehat u}}_1^{i-1}\right) \ge 0~{\rm and}~i\in{\mathcal A}\\
1,~~~~~{\rm if}~\theta_N^{\left(i\right)}\left({\bf y}_1^N, {\bf {\widehat u}}_1^{i-1}\right) < 0~{\rm and}~i\in{\mathcal A}\\
0,~~~~~{\rm if}~i\in{\mathcal F}.
\end{array} \right.
\end{equation}

\subsection{SCL Decoding}

The SCL decoding converts the SC decoding into a breadth-first search with the list size $L$ under the decoding complexity $O\left(LN\log N\right)$.
Given ${\mathcal U}_{i} \subseteq \left\{0,1\right\}^{i}, i \in [\![N]\!]$ be the subset including the $L$ survival paths at the $i$-th decoding step with $\left|{\mathcal U}_{i}\right| = L$, the probability of the survival path ${\bf {\widehat u}}_1^i \in {\mathcal U}_{i}$ is
\begin{equation}
P\left(\left.{\bf {\widehat u}}_1^i\right|{\bf y}_1^N\right) = \prod\limits_{j = 1}^i {P\left( {\left. {{{\widehat u}_j}} \right|{\bf y}_1^N,{\bf{\widehat u}}_1^{j - 1}} \right)}.
\end{equation}
When $u_{i}$ is an information bit, the $L$ survival paths in ${\mathcal U}_{i-1}$ are split into $2L$ paths with attempting calculating $P(\left.{\widehat {\bf u}}_1^{i-1}, {\widehat u}_{i}=0\right|{\bf y}_1^N)$ and $P(\left.{\widehat {\bf u}}_1^{i-1}, {\widehat u}_{i}=1\right|{\bf y}_1^N)$.
Then, ${\mathcal U}_{i}$ is decided by selecting the $L$ most likely paths from the $2L$ split paths. When $u_{i}$ is frozen bit, the $L$ survival paths in ${\mathcal U}_{i-1}$ are simply extended with the correct frozen bit. After the $N$-th decoding step, the most likely path in ${\mathcal U}_N$ is selected as the decoded path.

\subsection{Performance Upper Bound of SC Decoding}

Given the information set ${\mathcal A} \subseteq [\![N]\!]$, the block error event of the SC decoding, denoted by ${\mathcal E}$, is a union over the events that the first decision error occurs at the $i$-th decoding step.
Hence, we have
\begin{equation}
{\mathcal E} = \bigcup\limits_{i \in {\mathcal A}} {\mathcal D}_i,
\end{equation}
where
\begin{equation}
\begin{aligned}
{\mathcal D}_i &= \left\{ {\bf y}_1^N, {\bf{\widehat u}}_1^{i-1} = {\bf u}_1^{i-1},  {{\widehat u}}_i \ne u_i \right\} \\
&\subset
{\left\{ {\bf y}_1^N, {{\widehat u}}_i \ne u_i \right\}}{= :{\mathcal E}_i}.
\end{aligned}
\end{equation}
Then, since
${\mathcal E} \subset \bigcup\limits_{i \in {\mathcal A}} {\mathcal E}_i$,
the performance of SC decoding is upper bounded as
\begin{equation}\label{EqSCUpperBound}
P\left({\mathcal E}\right)
= \sum_{i \in {\mathcal A}}{P\left({\mathcal D}_i\right)}
\le \sum_{i \in {\mathcal A}}{P\left({\mathcal E}_i\right)}
\le \sum_{i \in {\mathcal A}}{Z\left(W_N^{\left(i\right)}\right)},
\end{equation}
where $Z\left(W_N^{\left(i\right)}\right)$ is the Bhattacharyya parameter of $W_N^{\left(i\right)}$.

In DE/GA algorithm, $P\left({\mathcal E}_i\right)$ is equal to the error probability calculated by the probability density functions $p\left(\theta_N^{\left(i\right)}\left({\bf y}_1^N, {\bf {\widehat u}}_1^{i-1} = {\bf 0}_1^{i-1}\right)\right)$ with the case ${\bf u}_1^N = {\bf 0}_1^N$. By discarding ${\bf y}_1^N$ and ${\bf {u}}_1^{j-1}$ we have
\begin{equation}
P\left({\mathcal E}_i\right) = P\left({\mathsf \Theta}_N^{\left(i\right)} < 0\right) = \int\limits_{- \infty}^0 {p\left(\theta_N^{\left(i\right)}\right)d \theta_N^{\left(i\right)}},
\end{equation}
where ${\mathsf \Theta}_N^{\left(i\right)}$ is the random variable of $\theta_N^{\left(i\right)}$.

\subsection{Simplified Code Tree of Polar Codes}

In this section, we describe the code tree of polar codes in \cite{ImprovedSCLniu} and prune the branches of frozen bit to decrease the number of the levels of code tree from $N$ to $K$.

The SCL decoding can be represented as a breadth-first path search on a perfect binary tree with $N$ levels.
A path from the root node to an $i$-level node represents an information subvector ${\bf u}_1^i \in \left\{0,1\right\}^i$ with the probability
\begin{equation}\label{EqProPathN}
P\left(\left.{\bf u}_1^i\right|{\bf y}_1^N\right) = \prod\limits_{j \in [\![i]\!]}{P\left(\left.u_j\right|{\bf y}_1^N, {\bf u}_1^{j-1}\right)}.
\end{equation}
Then, for the information bits, given the LLR defined in \eqref{EqLLR} and $\frac{P(u_j = 0)}{P(u_j = 1)} = 1$, we have
\begin{equation}\label{EqLLRCodeTree}
\begin{aligned}
\theta_N^{\left(j\right)}\left({\bf y}_1^N, {\bf {u}}_1^{j-1}\right)
&= \ln\frac{W_N^{\left(j\right)}\left(\left.{\bf y}_1^N, {\bf {u}}_1^{j-1}\right|u_j = 0\right)}{W_N^{\left(j\right)}\left(\left.{\bf y}_1^N, {\bf {u}}_1^{j-1}\right|u_j = 1\right)}\\
&= \ln\frac{P\left(\left.u_j = 0\right|{\bf y}_1^N, {\bf {u}}_1^{j-1}\right)}
{P\left(\left.u_j = 1\right|{\bf y}_1^N, {\bf {u}}_1^{j-1}\right)}.
\end{aligned}
\end{equation}
Discarding ${\bf y}_1^N$ and ${\bf {u}}_1^{j-1}$ from $\theta_N^{\left(j\right)}\left({\bf y}_1^N, {\bf {u}}_1^{j-1}\right)$, we can obtain
\begin{equation}
P\left(\left.u_j\right|{\bf y}_1^N, {\bf {u}}_1^{j-1}\right) =
\left\{ \begin{aligned}
&\left(e^{-\left(1-2u_j\right)\theta_N^{\left(j\right)}} + 1\right)^{-1}&,{\rm if}~j \in {\mathcal A},\\
&1&,{\rm if}~j \in {\mathcal F}.
\end{aligned} \right.
\end{equation}
With this, \eqref{EqProPathN} is simplified as
\begin{equation}
P\left(\left.{\bf u}_1^i\right|{\bf y}_1^N\right) = \prod\limits_{j\in\left({\mathcal A}\bigcap [\![i]\!]\right)}
{P\left(\left.u_j\right|{\bf y}_1^N, {\bf u}_1^{j-1}\right)}
\end{equation}
and there are $2^K$ paths in the binary tree with $N$ levels.
Assuming $u_j = 0$ for $j \in {\mathcal F}$, the $j$-level branch with the frozen bit equal to 1 is pruned.
Hence, the code tree of $\left(N, K\right)$ polar code is pruned as a perfect binary tree with $K$ levels and each path from the root node to the $K$-level node represents the information sequence ${\bf u}_1^N = \left({\bf u}_{\mathcal A}, {\bf u}_{\mathcal F} = 0\right)$ with ${\bf u}_{\mathcal A} \in \left\{0,1\right\}^K$.

\begin{figure}[t]
\setlength{\abovecaptionskip}{0.cm}
\setlength{\belowcaptionskip}{-0.cm}
  \centering{\includegraphics[scale=0.9]{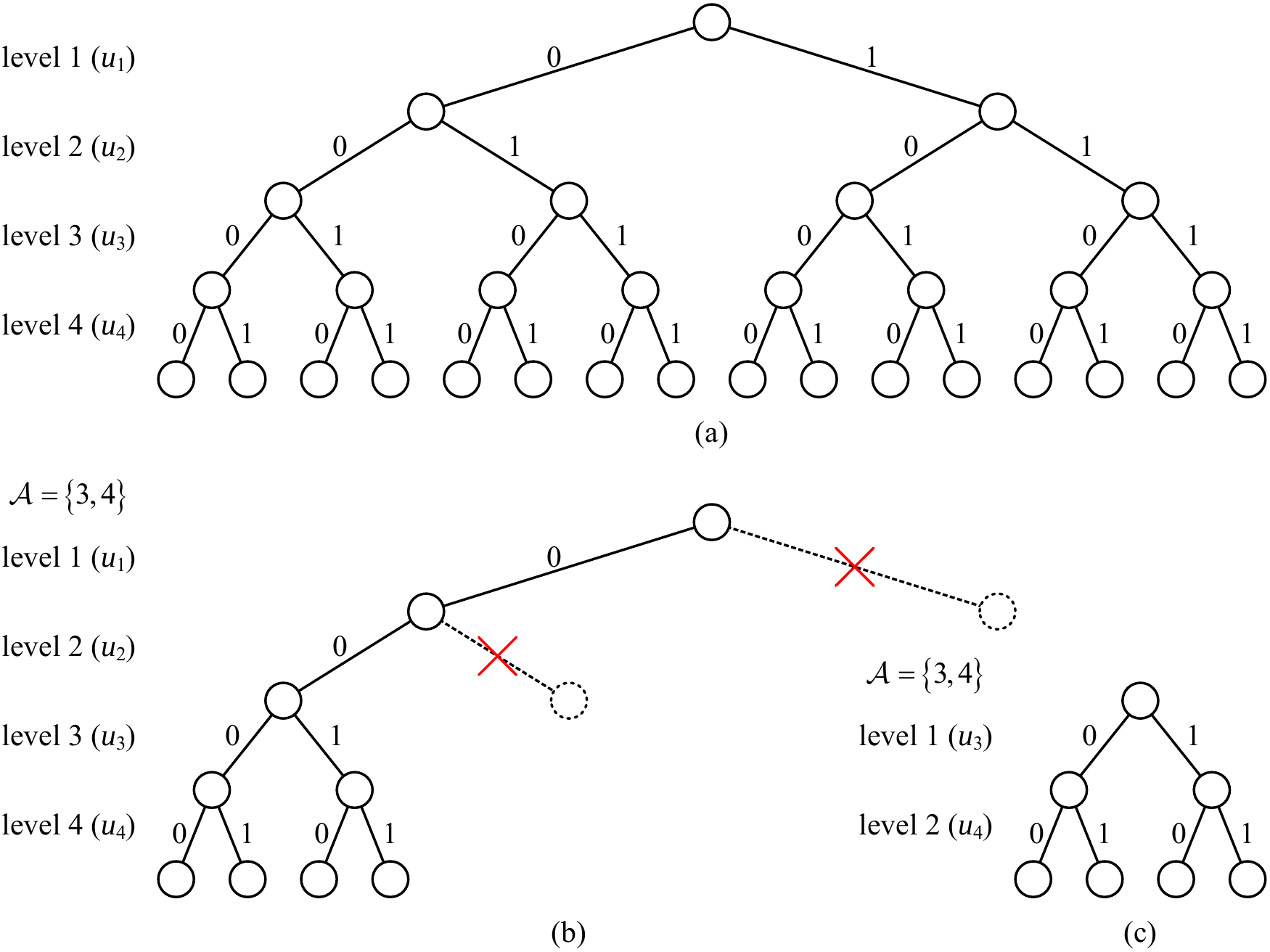}}
  \caption{An example of the code tree with $N = 4$, $K = 2$ and ${\mathcal A} = \left\{3,4\right\}$.}\label{FigCodeTree}
  \vspace{-0em}
\end{figure}

Fig. \ref{FigCodeTree} illustrates an example of the code tree with $N = 4$, $K = 2$ and ${\mathcal A} = \left\{3,4\right\}$. Fig. \ref{FigCodeTree}(a) is a perfect binary tree with $N=4$ levels. Then, pruning the branches with the frozen bit equal to 1 is shown in Fig. \ref{FigCodeTree}(b). After the pruning, we obtain another perfect binary tree with $K=2$ levels in Fig. \ref{FigCodeTree}(c), where level 1 and level 2 are related to $u_3$ and $u_4$, respectively.

\section{An Outline of Performance Analysis of SCL Decoding}

In this section, we first divide the block error event of SCL decoding into the PL error event and the PS error event. Then, we provide the bounds of the block error event, the PL error event and the PS error event.

\subsection{Block Error Event of SCL Decoding}

For an $\left(N, K\right)$ polar code with code rate $R = K/N$, the block error event ${\mathcal E}_{\tt SCL}$ of the SCL decoding is composed of the PL error event ${\mathcal E}_{\tt PL}$ and the PS error event ${\mathcal E}_{\tt PS}$, i.e.,
\begin{equation}
{\mathcal E}_{\tt SCL} = {\mathcal E}_{\tt PL} \bigcup {\mathcal E}_{\tt PS},
\end{equation}
where
${\mathcal E}_{\tt PL}$ is that the correct path is lost during the SCL decoding and ${\mathcal E}_{\tt PS}$ is that the correct path is reserved in the decoded list after SCL decoding but not selected as the decoded codeword by the path metric. Thus, we have
\begin{equation}
P\left({\mathcal E}_{\tt SCL}\right) = P\left({\mathcal E}_{\tt PL}\right) + P\left({\mathcal E}_{\tt PS}\right).
\end{equation}

\subsection{Lower Bound of $P\left({\mathcal E}_{\tt PL}\right)$}

${\mathcal E}_{\tt PL}$ is a union over the events that the correct path is first lost at the $i$-th decoding step, i.e.,
\begin{equation}\label{EqErrorPL}
{\mathcal E}_{\tt PL} = \bigcup\limits_{i \in [\![N]\!]} {\mathcal P}_i,
\end{equation}
where
\begin{equation}\label{EqPiDef}
\begin{aligned}
{\mathcal P}_i
&= \left\{ {\bf y}_1^N,
{\bf u}_1^{j} \in {\mathcal U}_{j}, j\in[\![i-1]\!],  {\bf u}_1^{i} \notin {\mathcal U}_{i}\right\} \\
&= \left\{ {\bf y}_1^N,
{\bf u}_1^{i-1} \in {\mathcal U}_{i-1}, {\bf u}_1^{i} \notin {\mathcal U}_{i}\right\}
\end{aligned}
\end{equation}
and ${\mathcal U}_{i}$ is the list set including the $L$ survival paths at the $i$-th decoding step.

Given $m = \log L$, we have
\begin{equation}\label{EqPiEmptyList}
{\mathcal P}_i = \varnothing, {\rm if}~i\in \left\{a_1,a_2,\cdots,a_m\right\},
\end{equation}
where $a_k$, $k \in [\![K]\!]$, is the $k$-th least element in ${\mathcal A}$, i.e.,
${\mathcal A} = \left\{a_1,a_2,\cdots,a_K\right\}$ with $a_1 < a_2 < \cdots < a_K$.
For the index in the frozen set ${\mathcal F}$, we also have
\begin{equation}\label{EqPiEmptyFrozen}
{\mathcal P}_i = \varnothing, {\rm if}~i\in {\mathcal F}.
\end{equation}
Since the frozen bits have no influence on ${\mathcal E}_{\tt PL}$, we introduce another information vector ${\bf v}_1^K = \left(v_1,v_2,\cdots,v_K\right)$ with $v_k = u_{a_k}$, $k \in [\![K]\!]$.
By \eqref{EqPiEmptyList} and \eqref{EqPiEmptyFrozen}, \eqref{EqErrorPL} is simplified as
\begin{equation}
{\mathcal E}_{\tt PL} = \bigcup\limits_{i \in [\![N]\!]} {\mathcal P}_i =
\bigcup\limits_{k = m+1}^{K} {\mathcal P}_{a_k} = \bigcup\limits_{k = m+1}^{K} {\mathcal Q}_{k},
\end{equation}
where
\begin{equation}
{\mathcal Q}_{k} = \left\{ {\bf y}_1^N,
{\bf v}_1^{k-1} \in {\mathcal V}_{k-1},  {\bf v}_1^{k} \notin {\mathcal V}_{k}\right\} \end{equation}
and ${\mathcal V}_{k}$ is the list set including the $L$ survival paths at the $a_k$-th decoding step.

Then, we introduce the assumption that ${\bf v}_1^K = {\bf 0}_1^K$, which is also used to calculate the upper bound of SC decoding in DE/GA algorithm. With the assumption, the lower bound of $P\left({\mathcal E}_{\tt PL}\right)$ is derived in Lemma \ref{LemmaLowerBoundOutline}.

\begin{lemma}\label{LemmaLowerBoundOutline}
Assuming ${\bf v}_1^K = {\bf 0}_1^K$, the lower bound of $P\left({\mathcal E}_{\tt PL}\right)$ is
\begin{equation}\label{EqPplDef}
P\left({\mathcal E}_{\tt PL}\right) =  \sum\limits_{k = m+1}^{K} P\left({\mathcal Q}_{k}\right) \ge \sum\limits_{k = m+1}^{K} P\left({\mathcal S}_{k}\right),
\end{equation}
where
\begin{equation}
{\mathcal S}_{k} = \left\{ {\bf y}_1^N,
{\mathcal V}_{k-1} = {\mathcal Z}_{k-1}^m,  {\bf 0}_1^{k} \notin {\mathcal V}_{k}\right\}
\end{equation}
is a special case of ${\mathcal Q}_{k}$ with
\begin{equation}
{\mathcal Z}_{k-1}^m=
{\left\{ \left. \left({\bf {\widehat v}}_1^{k-m-1} = {\bf 0}_1^{k-m-1}, {\bf {\widehat v}}_{k-m}^{k-1}\right) \right| {\bf {\widehat v}}_{k-m}^{k-1} \in \left\{0,1\right\}^m \right\}}
\end{equation}
which includes $2^m$ survival paths where the first $\left(k-m-1\right)$ information bits are equal to 0 and the subsequent $m$ information bits are ergodic.
\end{lemma}
\begin{proof}
Due to ${\mathcal S}_{k} \subset {\mathcal Q}_{k}$, the proof of Lemma \ref{LemmaLowerBoundOutline} is clear.
\end{proof}

Furthermore, $P\left({\mathcal S}_{k}\right)$ is represented as
\begin{equation} \label{EqPSkLowerBound}
\begin{aligned}
&P\left({\mathcal S}_{k}\right)\\
&\overset{(a)}{=} P\left(\left.{\bf 0}_1^{k} \notin {\mathcal V}_{k}\right|{\bf y}_1^N,
{\mathcal V}_{k-1} = {\mathcal Z}_{k-1}^m\right)P\left( {\bf y}_1^N,
{\mathcal V}_{k-1} = {\mathcal Z}_{k-1}^m \right)   \\
&\overset{(b)}{=} P\left(\left.{\bf 0}_1^{k} \notin {\mathcal V}_{k}\right|{\bf y}_1^N,
{\mathcal V}_{k-1} = {\mathcal Z}_{k-1}^m\right)
\underbrace{
P\left( \left.\left\{{\bf {\widehat v}}_{k-m}^{k-1} \in \left\{0,1\right\}^m \right\} \right| {\bf y}_1^N, {\bf {\widehat v}}_1^{k-m-1} = {\bf 0}_1^{k-m-1} \right)
}_{=1}
P\left( {\bf y}_1^N, {\bf {\widehat v}}_1^{k-m-1} = {\bf 0}_1^{k-m-1} \right) \\
&\overset{(c)}{=} P\left(\left.{\bf 0}_1^{k} \notin {\mathcal V}_{k}\right|{\bf y}_1^N,
{\mathcal V}_{k-1} = {\mathcal Z}_{k-1}^m\right)
\prod_{j=1}^{k-m-1}
{P\left({\bf y}_1^N, \left.{\widehat v}_j = 0 \right|  {\bf {\widehat v}}_1^{j-1} = {\bf 0}_1^{j-1}\right)}\\
&\overset{(d)}{=} P\left(\left.{\bf 0}_1^{k} \notin {\mathcal V}_{k}\right|{\bf y}_1^N, {\mathcal V}_{k-1} = {\mathcal Z}_{k-1}^m\right)
\prod_{j=1}^{k-m-1}{P\left({\mathsf L}_j \ge 0\right)} \\
&\overset{(e)}{=}\left(1 - P\left(\underbrace{\left.{\bf 0}_1^{k} \in {\mathcal V}_{k}\right|{\bf y}_1^N, {\mathcal V}_{k-1} = {\mathcal Z}_{k-1}^m}_{=:{\mathcal C}_k}\right)\right)
\prod_{j=1}^{k-m-1}{\int\limits_{0}^{\infty} {p\left(l_j\right)d l_j}}\\
&\overset{(f)}{=}\left(1 - P\left({\mathcal C}_k\right)\right)
\prod_{j=1}^{k-m-1}{\int\limits_{0}^{\infty} {p\left(l_j\right)d l_j}}.
\end{aligned}
\end{equation}
In \eqref{EqPSkLowerBound}, $(a)$, $(b)$ and $(c)$ are derived by Bayes formula,
where $P\left( \left.\left\{{\bf {\widehat v}}_{k-m}^{k-1} \in \left\{0,1\right\}^m \right\} \right| {\bf y}_1^N, {\bf {\widehat v}}_1^{k-m-1} = {\bf 0}_1^{k-m-1} \right)$ in $(b)$ is equal to 1,
since all the $m$ bits in ${\bf {\widehat v}}_{k-m}^{k-1}$ are ergodic for the $L = 2^m$ paths of ${\bf{\widehat v}}_1^{k-1}$ when ${\bf {\widehat v}}_1^{k-m-1} = {\bf 0}_1^{k-m-1}$.
To derive $(d)$ and $(e)$, given $\frac{P\left({\widehat v}_j = 0\right)}{P\left({\widehat v}_j = 1\right)} = 1$ and
the random variable ${\mathsf L}_j$ of the LLR $l_j$ defined as
\begin{equation}\label{EqDeflj}
\begin{aligned}
l_j
    &=\ln\frac
        {p\left({\bf y}_1^N, \left.{v}_j = 0 \right|  {\bf {\widehat v}}_1^{j-1} = {\bf 0}_1^{j-1}\right)}
        {p\left({\bf y}_1^N, \left.{v}_j = 1 \right|  {\bf {\widehat v}}_1^{j-1} = {\bf 0}_1^{j-1}\right)}\\
    &=\ln\frac
        {W_N^{\left({a_j}\right)}\left(\left.{\bf y}_1^N, {\bf {\widehat u}}_1^{{a_j}-1}= {\bf 0}_1^{a_j-1}\right|u_{a_j} = 0\right)}
        {W_N^{\left({a_j}\right)}\left(\left.{\bf y}_1^N, {\bf {\widehat u}}_1^{{a_j}-1}= {\bf 0}_1^{a_j-1}\right|u_{a_j} = 1\right)}\\
    &= \theta_N^{\left({a_j}\right)}\left({\bf y}_1^N, {\bf {\widehat u}}_1^{{a_j}-1}= {\bf 0}_1^{a_j-1}\right)
\end{aligned}
\end{equation}
with the probability density function $p\left(l_j\right)$, we have
\begin{equation}
P\left({\bf y}_1^N, \left.{\widehat v}_j = 0 \right|  {\bf {\widehat v}}_1^{j-1} = {\bf 0}_1^{j-1}\right) = P\left({\mathsf L}_j \ge 0\right) = \int\limits_{0}^{\infty} {p\left(l_j\right)d l_j},
\end{equation}
where $p\left(l_j\right) = p\left(\theta_N^{\left({a_j}\right)}\right)$ can be obtained by DE/GA algorithm.
Then, in $(e)$ and $(f)$, given ${\mathcal C}_k := \left\{\left.{\bf 0}_1^{k} \in {\mathcal V}_{k}\right|{\bf y}_1^N, {\mathcal Z}_{k-1}^m\right\}$,
we replace $P\left(\left.{\bf 0}_1^{k} \notin {\mathcal V}_{k}\right|{\bf y}_1^N, {\mathcal Z}_{k-1}^m\right)$ with $\left(1 - P\left({\mathcal C}_k\right)\right)$.
In Section IV, the upper bound of $P\left({\mathcal C}_k\right)$ is derived to obtain the lower bound of $P\left({\mathcal E}_{\tt PL}\right)$.

\subsection{Approximation of $P\left({\mathcal E}_{\tt PS}\right)$}

${\mathcal E}_{\tt PS}$ is represented as
\begin{equation}
\begin{aligned}
{\mathcal E}_{\tt PS}
    &=\left\{
        \begin{aligned}
        {\bf y}_1^N, {\bf v}_1^K \in {\mathcal V}_K,
        \exists {\bf{\widehat v}}_1^K \in {\mathcal V}_K, P\left(\left.{\bf v}_1^K\right|{\bf y}_1^N\right)<P\left(\left.{\bf {\widehat v}}_1^K\right|{\bf y}_1^N\right)
        \end{aligned}
    \right\}\\
    &\subset \left\{
        {\bf y}_1^N, \exists {\bf{\widehat v}}_1^K \in \left\{0,1\right\}^K, P\left(\left.{\bf v}_1^K\right|{\bf y}_1^N\right)<P\left(\left.{\bf {\widehat v}}_1^K\right|{\bf y}_1^N\right)
    \right\}.
\end{aligned}
\end{equation}
Hence, $P\left({\mathcal E}_{\tt PS}\right)$ is bounded by the ML performance $P_{\tt ML}$, i.e.,
\begin{equation}
P\left({\mathcal E}_{\tt PS}\right) \le P_{\tt ML}.
\end{equation}
Then, since $P\left({\mathcal E}_{\tt PS}\right)$ is approximated to $P_{\tt ML}$ for $L \ge 2$ in the high SNR region, we have
\begin{equation}\label{EqPpsML}
P\left({\mathcal E}_{\tt PS}\right) \approx P_{\tt ML}, {\rm if}~L\ge2~{\rm and}~\frac{E_s}{N_0}\gg 1.
\end{equation}

\subsection{Lower Bound of $P\left({\mathcal E}_{\tt SCL}\right)$}

The lower bound of $P\left({\mathcal E}_{\tt SCL}\right)$ is provided in Theorem \ref{TheoremLowerBoundOutline}.
\begin{theorem}\label{TheoremLowerBoundOutline}
The lower bound of $P\left({\mathcal E}_{\tt SCL}\right)$ is
\begin{equation}\label{EqSCLLowerBoundOutline}
\begin{aligned}
P\left({\mathcal E}_{\tt SCL}\right)
= P\left({\mathcal E}_{\tt PL}\right) + P\left({\mathcal E}_{\tt PS}\right)
\gtrsim \sum\limits_{k = m+1}^{K} P\left({\mathcal S}_{k}\right) + P_{\tt ML}.
\end{aligned}
\end{equation}
\end{theorem}
\begin{proof}
The proof is clear by Lemma \ref{LemmaLowerBoundOutline} and \eqref{EqPpsML}.
\end{proof}

\section{Performance Lower Bound of SCL Decoding}

In this section, we first modify the performance upper bound of SC decoding by the performance analysis above. Then, we derive the performance lower bound of SCL decoding.

\subsection{Modified Upper Bound of SC Decoding}

The modified upper bound of SC decoding is provided in Proposition \ref{PropositionModifiedSCBound}.

\begin{proposition}\label{PropositionModifiedSCBound}
Assuming ${\bf v}_1^K = {\bf 0}_1^K$, the modified upper bound of SC decoding is
\begin{equation}\label{EqModifiedUpperBoundSC}
\begin{aligned}
P\left({\mathcal E}_{\tt SC}\right)
= \sum\limits_{k = 1}^{K} P\left({\mathcal S}_{k}\right)
= \sum\limits_{k = 1}^{K} \left(P\left({\mathsf L}_k < 0\right)
    \prod_{j=1}^{k-1}{P\left({\mathsf L}_j \ge 0\right)}\right).
\end{aligned}
\end{equation}
\end{proposition}
\begin{proof}
    SC decoding can be regarded as the SCL decoding with $L = 1$ and $m = 0$.
    Hence, we have
    \begin{equation}\label{EqSCBoundPps}
    P\left({\mathcal E}_{\tt PS}\right) = 0.
    \end{equation}
    Assuming ${\bf v}_1^K = {\bf 0}_1^K$, we have
    \begin{equation}
    {\mathcal Q}_{k} = \left\{ {\bf y}_1^N,
    {\bf v}_1^{k-1} = {\bf 0}_1^{k-1},  {v}_k \neq 0\right\} = {\mathcal S}_{k}.
    \end{equation}
    With \eqref{EqPSkLowerBound}, $P\left({\mathcal S}_{k}\right)$ is
    \begin{equation}\label{EqSCBoundPSk}
    \begin{aligned}
    P\left({\mathcal S}_{k}\right)
        &=P\left(\left.{v}_k \neq 0\right|{\bf y}_1^N, {\bf v}_1^{k-1} = {\bf 0}_1^{k-1}\right)
        \prod_{j=1}^{k-1}{P\left({\mathsf L}_j \ge 0\right)} \\
        &= P\left({\mathsf L}_k < 0\right)
        \prod_{j=1}^{k-1}{P\left({\mathsf L}_j \ge 0\right)}.
    \end{aligned}
    \end{equation}
    Thus, according to \eqref{EqPplDef}, \eqref{EqSCBoundPps} and \eqref{EqSCBoundPSk}, the modified upper bound of SC decoding is
    \begin{equation}
    \begin{aligned}
    P\left({\mathcal E}_{\tt SC}\right)
    &= P\left({\mathcal E}_{\tt PL}\right) + P\left({\mathcal E}_{\tt PS}\right) \\
    &= \sum\limits_{k = 1}^{K} \left(P\left({\mathsf L}_k < 0\right)
        \prod_{j=1}^{k-1}{P\left({\mathsf L}_j \ge 0\right)}\right).
    \end{aligned}
    \end{equation}
\end{proof}

\begin{remark}
Clearly, due to $P\left({\mathcal E}_{\tt SC}\right) \le \sum_{i \in {\mathcal A}}{P\left({\mathcal E}_i\right)}$, $P\left({\mathcal E}_{\tt SC}\right)$ is a tighter upper bound of SC decoding and we have $P\left({\mathcal E}_{\tt SC}\right) \approx \sum_{i \in {\mathcal A}}{P\left({\mathcal E}_i\right)}$ in the high SNR region.
\end{remark}

\subsection{Lower Bound of SCL Decoding with $L = 2$}

By \eqref{EqPplDef}, \eqref{EqPSkLowerBound} and \eqref{EqSCLLowerBoundOutline}, we should derive the upper bound of $P\left({\mathcal C}_k\right)$ to obtain the lower bound of $P\left({\mathcal E}_{\tt PL}\right)$ in order to determine the lower bound of $P\left({\mathcal E}_{\tt SCL}\right)$.

\begin{figure}[t]
\setlength{\abovecaptionskip}{0.cm}
\setlength{\belowcaptionskip}{-0.cm}
  \centering{\includegraphics[scale=0.9]{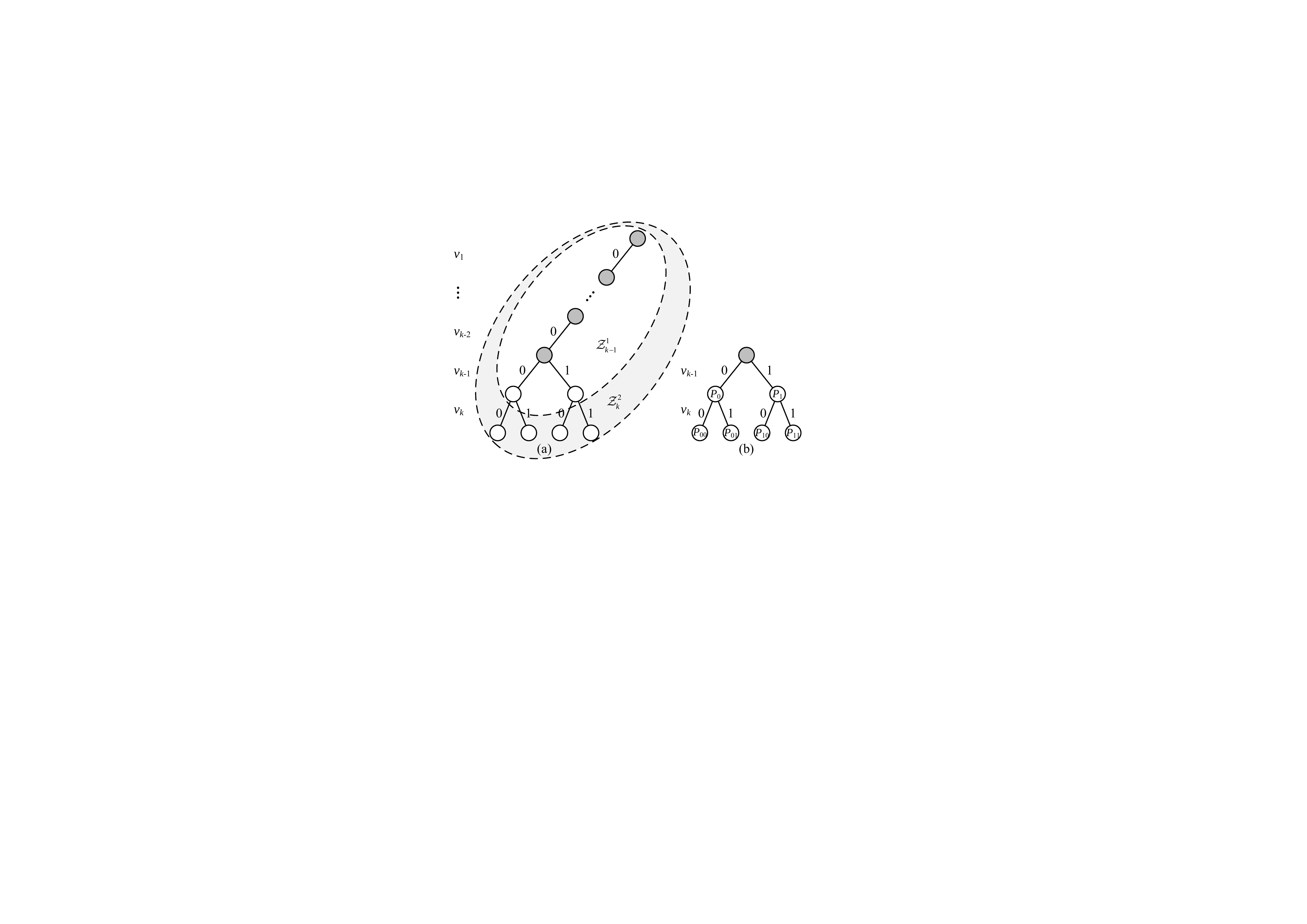}}
  \caption{The subtrees with $L = 2$. Fig. \ref{FigCodeTreeZ}(a) illustrates ${\mathcal Z}_{k-1}^1$ and ${\mathcal Z}_{k}^2$. Fig. \ref{FigCodeTreeZ}(b) shows the $4$ paths $\left(v_{k-1},v_k\right) \in \left\{0,1\right\}^2$ and corresponding path metric by discarding ${\bf v}_{1}^{k-2}$.}\label{FigCodeTreeZ}
  \vspace{-0em}
\end{figure}

For SCL decoding with $L=2$ and $m=1$, we have
\begin{equation}
{\mathcal Z}_{k-1}^1 =
\left\{
\left({\bf 0}_1^{k-2}, v_{k-1} = 0 \right),\left({\bf 0}_1^{k-2}, v_{k-1} = 1 \right)
\right\}
\end{equation}
and
\begin{equation}
{\mathcal Z}_{k}^2 =
\left\{
\begin{aligned}
    \left({\bf 0}_1^{k-2}, v_{k-1}^{k} = (0,0), \right),
    \left({\bf 0}_1^{k-2}, v_{k-1}^{k} = (0,1), \right)\\
    \left({\bf 0}_1^{k-2}, v_{k-1}^{k} = (1,0), \right),
    \left({\bf 0}_1^{k-2}, v_{k-1}^{k} = (1,1), \right)
\end{aligned}
\right\},
\end{equation}
where the $4$ paths in ${\mathcal Z}_{k}^2$ is the extended paths of the $2$ paths in ${\mathcal Z}_{k-1}^1$.
${\mathcal Z}_{k-1}^1$ and ${\mathcal Z}_{k}^2$ are illustrated in Fig. \ref{FigCodeTreeZ}(a).

When ${\mathcal V}_{k-1} = {\mathcal Z}_{k-1}^1$, we have ${\mathcal V}_{k} \subset {\mathcal Z}_{k}^2$ including $2$ paths in ${\mathcal Z}_{k}^2$ with the largest path metric.
With this, ${\mathcal C}_k$ is equivalent to the event that $P\left(\left.{\bf 0}_1^k\right|y_1^N\right)$ is at least the second largest path metric among the paths in ${\mathcal Z}_{k}^2$, i.e.,
\begin{equation}\label{EqSCL2ProbSetInitial}
{\mathcal C}_k = \left\{\left.P\left(\left.{\bf 0}_1^k\right|{\bf y}_1^N\right) \ge \max\nolimits^2_{{\bf v}_1^k \in {\mathcal Z}_{k}^2} P\left(\left.{\bf v}_1^k\right|{\bf y}_1^N\right)  \right| {\bf y}_1^N, {\mathcal V}_{k-1} = {\mathcal Z}_{k-1}^1 \right\}.
\end{equation}
Discarding ${\bf y}_1^N$ and ${\bf v}_1^{k-2}$ from \eqref{EqSCL2ProbSetInitial}, ${\mathcal C}_k$ is simplified as
\begin{equation}
{\mathcal C}_k = \left\{P_{00} \ge \max\nolimits^2\left\{P_{00}, P_{01}, P_{10}, P_{11}\right\}\right\},
\end{equation}
where
\begin{equation}
\begin{aligned}
P_{v_{k-1},v_k}
    &= P\left(\left.v_{k-1},v_k\right|{\bf y}_1^N, {\bf v}_1^{k-2} = {\bf 0}_1^{k-2} \right) \\
    &= P\left(\left.v_{k-1}\right|{\bf y}_1^N, {\bf v}_1^{k-2} = {\bf 0}_1^{k-2} \right)
    P\left(\left.v_k\right|{\bf y}_1^N, {\bf v}_1^{k-2} = {\bf 0}_1^{k-2}, v_{k-1} \right),
\end{aligned}
\end{equation}
and $\max\nolimits^2\left\{P_{00}, P_{01}, P_{10}, P_{11}\right\}$ means the second largest path metric in $\left\{P_{00}, P_{01}, P_{10}, P_{11}\right\}$.
The subtree in Fig. \ref{FigCodeTreeZ}(b) shows the $4$ paths $\left(v_{k-1},v_k\right) \in \left\{0,1\right\}^2$ and the corresponding path metric.

Then, we provide Lemma \ref{LemmaProperty1} and Lemma \ref{LemmaProperty2} to describe the properties of $P_{v_{k-1},v_k}$ as follows.
\begin{lemma}\label{LemmaProperty1}
Given $P_{v_{k-1}} = P\left(\left.v_{k-1}\right|{\bf y}_1^N, {\bf 0}_1^{k-2} \right)$, we have
$
P_{v_{k-1}} = P_{v_{k-1},0} + P_{v_{k-1},1}.
$
\end{lemma}
\begin{lemma}\label{LemmaProperty2}
If $P_0 \ge P_1$, we have
$
\max\left\{P_{00}, P_{01}\right\} \ge \frac{P_0}{2} \ge \frac{P_1}{2} \ge \min\left\{P_{10}, P_{11}\right\}
$
and vice versa.
\end{lemma}

Due to DE/GA algorithm, the probability density functions of the LLRs $l_{k-1}$ and $l_{k}$ defined in \eqref{EqDeflj} are obtained. Hence, $P\left({\mathcal C}_k\right)$ is divided into three parts as
\begin{equation}\label{EqSCL2ThreeParts}
\begin{aligned}
P\left({\mathcal C}_k\right) =
&P\left( {{\mathsf L}_{k-1}}\ge 0,{{\mathsf L}_{k}}\ge 0, {\mathcal C}_k\right) + \\
&P\left( {{\mathsf L}_{k-1}}\ge 0,{{\mathsf L}_{k}}<0, {\mathcal C}_k\right) + \\
&P\left( {{\mathsf L}_{k-1}}<0, {\mathcal C}_k \right).
\end{aligned}
\end{equation}
The subtrees of the three parts are illustrated in Fig. \ref{FigCodeTreeSCL2}(a), (b) and (c), respectively, where the probability of the red bold node is equal or larger than that of the black node attached to the identical node and the dotted node represents its probability can not be calculated by $p\left(l_{k-1}\right)$ and $p\left(l_{k}\right)$.
For example, in Fig. \ref{FigCodeTreeSCL2}(a), we have $P_0 \ge P_1$ and $P_{00} \ge P_{01}$, which is shown in red nodes.
Then, in Fig. \ref{FigCodeTreeSCL2}(c),
since ${{\mathsf L}_{k-1}}<0$ does not meet the condition of $p\left(l_{k}\right)$,
where $p\left(l_{k}\right)$ is determined when ${\bf v}_1^K = {\bf 0}_1^K$ and ${\bf {\widehat v}}_1^{k-1} = {\bf 0}_1^{k-1}$, i.e.,
\begin{equation}
p\left(l_{k}\right) = p\left(\left.l_{k}\right|{\bf {\widehat v}}_1^{k-1} = {\bf 0}_1^{k-1}\right) = p\left(\left.l_{k}\right|{\mathsf L}_{j}\ge 0, j\in[\![k-1]\!]\right),
\end{equation}
we can not obtain $p\left(l_{k}\right)$, which leads to that $P_{00}, P_{01}, P_{10}$ and $P_{11}$ are indeterminate.

\begin{figure}[t]
\setlength{\abovecaptionskip}{0.cm}
\setlength{\belowcaptionskip}{-0.cm}
  \centering{\includegraphics[scale=0.9]{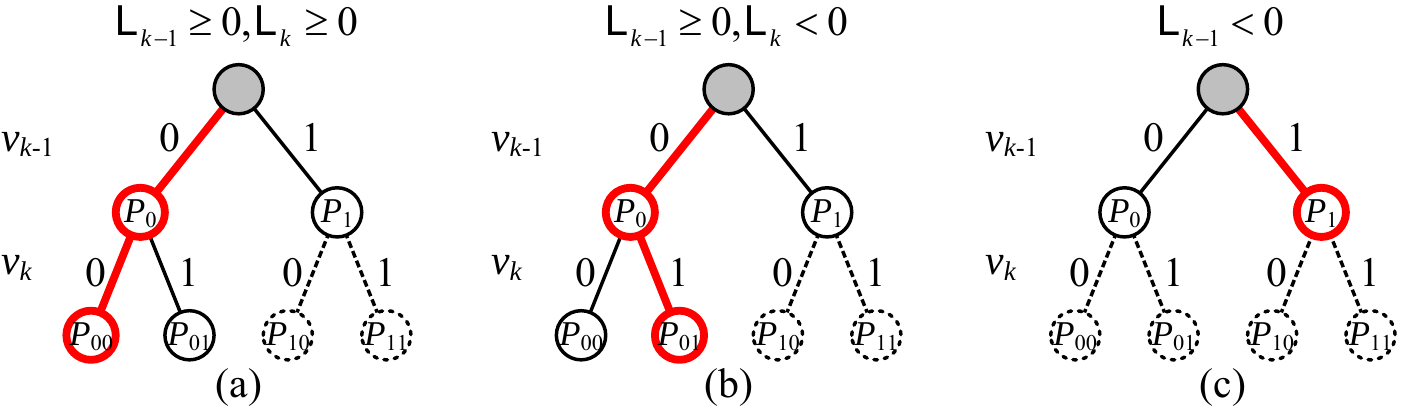}}
  \caption{The subtrees of the three divided parts of ${\mathcal C}_k$ with $L = 2$.}\label{FigCodeTreeSCL2}
  \vspace{-0em}
\end{figure}

Then, we derive the upper bound of the three parts to provide the upper bound of $P\left({\mathcal C}_k\right)$ in Proposition \ref{PropositionSCL2PCkUpperBound}.
\begin{proposition}\label{PropositionSCL2PCkUpperBound}
The upper bound of $P\left({\mathcal C}_k\right)$ is
\begin{equation}\label{EqPCkSCL2UpperBound}
\begin{aligned}
P\left({\mathcal C}_k\right) \le
&P\left( {{\mathsf L}_{k-1}}\ge 0,{{\mathsf L}_{k}}\ge 0\right) + \\
&P\left( {{\mathsf L}_{k-1}}\ge 0,-\ln \left( 2{{e}^{{{\mathsf L}_{k-1}}}}-1 \right)\le{{\mathsf L}_{k}}<0\right) + \\
&P\left( \ln \alpha \le{{\mathsf L}_{k-1}}<0 \right),
\end{aligned}
\end{equation}
where $\alpha = \frac{\min\left\{P_{10},P_{11}\right\}}{P_1}$.
\end{proposition}
\begin{proof}
For $P\left( {{\mathsf L}_{k-1}}\ge 0,{{\mathsf L}_{k}}\ge 0, {\mathcal C}_k\right)$ illustrated in Fig. \ref{FigCodeTreeSCL2}(a), due to $P_0 \ge P_1$ and $P_{00}\ge P_{01}$, by Lemma \ref{LemmaProperty2}, we have
\begin{equation}
\left\{
\begin{aligned}
&P_{00}\ge P_{01}, \\
&P_{00} = \max\left\{P_{00}, P_{01}\right\} \ge \min\left\{P_{10}, P_{11}\right\}.
\end{aligned}
\right.
\end{equation}
Hence, $P_{00}$ is at least the second largest path metric among $\left\{P_{00}, P_{01},P_{10}, P_{11}\right\}$ and we have
\begin{equation}\label{EqSCL2Part1Bound}
P\left( {{\mathsf L}_{k-1}}\ge 0,{{\mathsf L}_{k}}\ge 0, {\mathcal C}_k\right) = P\left( {{\mathsf L}_{k-1}}\ge 0,{{\mathsf L}_{k}}\ge 0\right).
\end{equation}

For $P\left( {{\mathsf L}_{k-1}}\ge 0,{{\mathsf L}_{k}}<0, {\mathcal C}_k\right)$ illustrated in Fig. \ref{FigCodeTreeSCL2}(b), due to $P_0 \ge P_1$ and $P_{00} < P_{01}$, we have
\begin{equation}\label{EqSCL2Part2Bound}
\begin{aligned}
P\left({{\mathsf L}_{k-1}}\ge 0,{{\mathsf L}_{k}}<0, {\mathcal C}_k\right)
&=  P\left( {{\mathsf L}_{k-1}}\ge 0,{{\mathsf L}_{k}}<0,
 P_{00}\ge\max \left\{P_{10}, P_{11}\right\}\right)\\
&\le P\left( {{\mathsf L}_{k-1}}\ge 0,{{\mathsf L}_{k}}<0,
 {{P}_{00}}\ge\frac{{{P}_{1}}}{2}\right)\\
&= P\left(
    \begin{aligned}
    &{{\mathsf L}_{k-1}}\ge 0,{{\mathsf L}_{k}}<0, \\
    &\frac{1}{{e^{-{\mathsf L}_{k-1}}}+1}\times \frac{1}{{e^{-{\mathsf L}_k}}+1} \ge \frac{1}{2}\times\frac{1}{{e^{{\mathsf L}_{k-1}}}+1}
    \end{aligned}
    \right)\\
&=P\left( {{\mathsf L}_{k-1}}\ge 0,-\ln \left( 2{{e}^{{\mathsf L}_{k-1}}}-1 \right)\le{{\mathsf L}_k}<0\right).
\end{aligned}
\end{equation}

For $P\left( {{\mathsf L}_{k-1}}<0, {\mathcal C}_k \right)$ illustrated in Fig. \ref{FigCodeTreeSCL2}(c), due to $P_0 < P_1$, we have
\begin{equation}\label{EqSCL2Part3Bound}
\begin{aligned}
&P\left( {{\mathsf L}_{k-1}}<0, {\mathcal C}_k \right) \\
    &= \underbrace{P\left( {{\mathsf L}_{k-1}}<0, {{P}_{00}}<{{P}_{01}},{{P}_{00}}\ge\max \left\{{{P}_{10}},{{P}_{11}}\right\} \right)}_{=0 {\rm~by~Lemma~\ref{LemmaProperty2}}} +
    P\left( {{\mathsf L}_{k-1}}<0, {{P}_{00}}\ge{{P}_{01}},{{P}_{00}}\ge\min \left\{{{P}_{10}},{{P}_{11}}\right\} \right) \\
    &\le  P\left( {{\mathsf L}_{k-1}}<0, {{P}_{00}}\ge\min \left\{{{P}_{10}},{{P}_{11}}\right\} \right) \\
    & \le P\left( {{\mathsf L}_{k-1}}<0, {{P}_{0}}\ge\min \left\{ {{P}_{10}},{{P}_{11}} \right\} \right) \\
    &= P\left( {{\mathsf L}_{k-1}}<0, {{P}_{0}}\ge \alpha P_1 \right) \\
    &= P\left( \ln \alpha \le{{\mathsf L}_{k-1}}<0 \right).
\end{aligned}
\end{equation}

Thus, by \eqref{EqSCL2Part1Bound} to  \eqref{EqSCL2Part3Bound}, the proof of Proposition \ref{PropositionSCL2PCkUpperBound} is completed.
\end{proof}

\subsection{Approximate Lower Bound of SCL Decoding with $L = 4$}

\begin{figure}[t]
\setlength{\abovecaptionskip}{0.cm}
\setlength{\belowcaptionskip}{-0.cm}
  \centering{\includegraphics[scale=0.9]{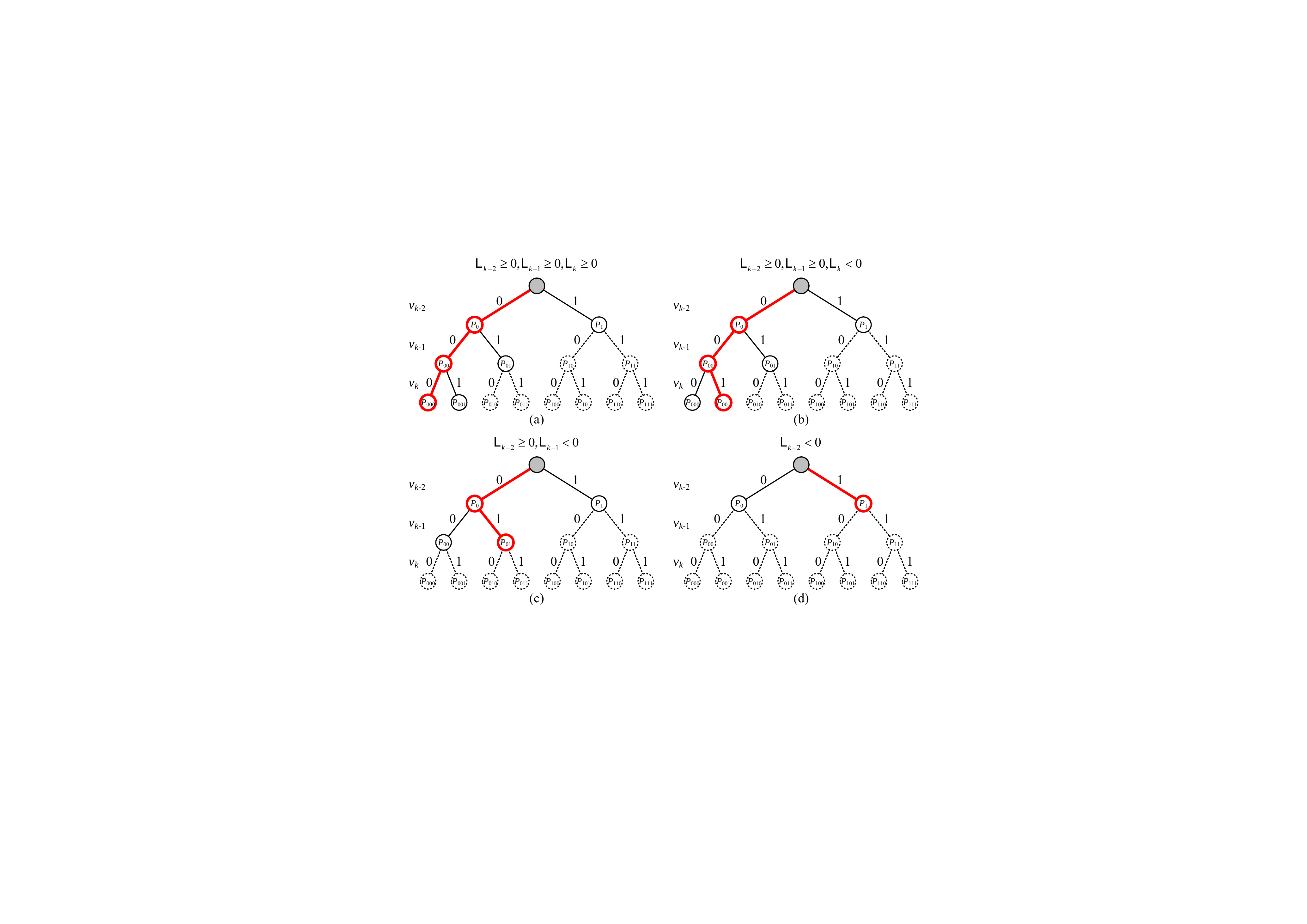}}
  \caption{The subtrees of the four divided parts of ${\mathcal C}_k$ with $L = 4$.}\label{FigCodeTreeSCL4}
  \vspace{-0em}
\end{figure}

For SCL decoding with $L = 4$ and $m = \log L = 2$, we have
\begin{equation}
{\mathcal C}_k = \left\{P_{000} \ge \underset{{{\bf v}_{k-2}^k \in \left\{0,1\right\}^3}}{\max\nolimits^4}\left\{P_{{\bf v}_{k-2}^k} \right\}\right\}
\end{equation}
and $P\left({\mathcal C}_k\right)$ is divided into four parts, i.e.,
\begin{equation}\label{EqSCL4FourParts}
\begin{aligned}
P\left({\mathcal C}_k\right) =
&P\left( {{\mathsf L}_{k-2}}\ge 0,{{\mathsf L}_{k-1}}\ge 0,{{\mathsf L}_{k}}\ge 0, {\mathcal C}_k\right) +
P\left( {{\mathsf L}_{k-2}}\ge 0, {{\mathsf L}_{k-1}}\ge 0,{{\mathsf L}_{k}}<0, {\mathcal C}_k\right) + \\
&P\left( {{\mathsf L}_{k-2}}\ge 0, {{\mathsf L}_{k-1}}<0, {\mathcal C}_k \right) +
P\left( {{\mathsf L}_{k-2}}< 0, {\mathcal C}_k \right).
\end{aligned}
\end{equation}

Fig. \ref{FigCodeTreeSCL4} illustrates the subtrees of the four parts in \eqref{EqSCL4FourParts}, where the probability of the red node is larger than that of the black node attached to the identical node and the probability of the dotted node is not determined.

Then, when ${{\mathsf L}_{k-2}}\ge 0$, since the correct path ${\bf 0}_{k-2}^k$ is not in $\left\{P_{100}, P_{101}, P_{110}, P_{110}\right\}$, we have the judgement of ${\bf v}_{k-1}^k$ attached $v_{k-2} = 1$ is random in $\left\{0,1\right\}^2$. Hence, the expectation of the 4 paths is
\begin{equation}\label{EqExpPSCL4}
E \left(P_{100}\right) = E \left(P_{101}\right) = E \left(P_{110}\right) = E \left(P_{111}\right) = \frac{P_1}{4}.
\end{equation}
With \eqref{EqExpPSCL4}, we propose Conjecture \ref{ConjectureSCL4} as follows.

\begin{conjecture}\label{ConjectureSCL4}
When ${{\mathsf L}_{k-2}}\ge 0$, we assume
\begin{equation}
P_{100} \approx P_{101} \approx P_{110} \approx P_{111} \approx \frac{P_1}{4}.
\end{equation}
\end{conjecture}

With Conjecture \ref{ConjectureSCL4}, the approximate upper bound of $P\left({\mathcal C}_k\right)$ is derived in Proposition \ref{PropositionSCL4PCkUpperBound}.

\begin{proposition}\label{PropositionSCL4PCkUpperBound}
The approximate upper bound of $P\left({\mathcal C}_k\right)$ is
\begin{equation}\label{EqPCkSCL4UpperBound}
\begin{aligned}
P\left({\mathcal C}_k\right) \lesssim
&P\left( {{\mathsf L}_{k-2}}\ge 0,{{\mathsf L}_{k-1}}\ge 0,{{\mathsf L}_{k}}\ge 0\right) +
P\left(
    \begin{aligned}
    &{{\mathsf L}_{k-2}}\ge 0,{{\mathsf L}_{k-1}} \ge 0,\\
    &-\ln \left( \frac{4e^{{\mathsf L}_{k-2}}}{e^{-{\mathsf L}_{k-1}} + 1}-1 \right)\le{{\mathsf L}_{k}}<0
    \end{aligned}
\right) + \\
&P\left(
    \begin{aligned}
    {{\mathsf L}_{k-2}}\ge 0,
    -\ln \left( 4{{e}^{{{\mathsf L}_{k-2}}}}-1 \right)\le{{\mathsf L}_{k-1}}<0
    \end{aligned}
\right) +
P\left( \ln \alpha \le{{\mathsf L}_{k-2}}<0 \right),
\end{aligned}
\end{equation}
where $\alpha = \frac{\min\left\{P_{100},P_{101},P_{110},P_{111}\right\}}{P_1}$.
\end{proposition}
\begin{proof}
For $P\left( {{\mathsf L}_{k-2}}\ge 0,{{\mathsf L}_{k-1}}\ge 0,{{\mathsf L}_{k}}\ge 0, {\mathcal C}_k\right)$ illustrated in Fig. \ref{FigCodeTreeSCL4}(a), by Lemma \ref{LemmaProperty2}, we have
\begin{equation}
\left\{
\begin{aligned}
&P_{000}\ge P_{001}, \\
&P_{000} = \max\left\{P_{000}, P_{001}\right\} \ge \min\left\{P_{010}, P_{011}\right\},\\
&P_{00} \ge \min\left\{P_{10}, P_{11}\right\} \Rightarrow
P_{000} = \max\left\{P_{000}, P_{001}\right\} \ge \min \left\{\begin{aligned}P_{100},P_{101},P_{110},P_{111}\end{aligned}\right\}.
\end{aligned}
\right.
\end{equation}
Hence, $P_{000}$ is at least the 5-th largest path metric among $\left\{P_{{\bf v}_{k-2}^k}, {\bf v}_{k-2}^k \in \left\{0,1\right\}^3\right\}$ and we have
\begin{equation}\label{EqSCL4Part1Bound}
\begin{aligned}
P\left( {{\mathsf L}_{k-2}}\ge 0,{{\mathsf L}_{k-1}}\ge 0,{{\mathsf L}_{k}}\ge 0, {\mathcal C}_k\right)
\le P\left( {{\mathsf L}_{k-2}}\ge 0,{{\mathsf L}_{k-1}}\ge 0,{{\mathsf L}_{k}}\ge 0\right).
\end{aligned}
\end{equation}

For $P\left( {{\mathsf L}_{k-2}}\ge 0, {{\mathsf L}_{k-1}}\ge 0,{{\mathsf L}_{k}}<0, {\mathcal C}_k\right)$ illustrated in Fig. \ref{FigCodeTreeSCL4}(b), by Conjecture \ref{ConjectureSCL4}, the probability is upper bounded as
\begin{equation}\label{EqSCL4Part2Transform}
\begin{aligned}
&P\left( {{\mathsf L}_{k-2}}\ge 0, {{\mathsf L}_{k-1}}\ge 0,{{\mathsf L}_{k}}<0, {\mathcal C}_k\right) \\
&\begin{aligned}
    = &P\left(
        {{\mathsf L}_{k-2}}\ge 0, {{\mathsf L}_{k-1}}\ge 0,{{\mathsf L}_{k}}<0,
         P_{000} \ge \max \left\{ {{P}_{010}},{{P}_{011}} \right\},
        P_{000} \ge \max\nolimits^3 \left\{P_{100},P_{101},P_{110},P_{111}\right\}
    \right) + \\
    &P\left(
        \begin{aligned}
        &{{\mathsf L}_{k-2}}\ge 0, {{\mathsf L}_{k-1}}\ge 0,{{\mathsf L}_{k}}<0,
         \min \left\{ {{P}_{010}},{{P}_{011}} \right\}\le P_{000} < \max \left\{ {{P}_{010}},{{P}_{011}} \right\},\\
        &P_{000} \ge \max\nolimits^2 \left\{P_{100},P_{101},P_{110},P_{111}\right\}
        \end{aligned}
    \right) + \\
    &P\left(
        {{\mathsf L}_{k-2}}\ge 0, {{\mathsf L}_{k-1}}\ge 0,{{\mathsf L}_{k}}<0,
         P_{000} < \min \left\{ {{P}_{010}},{{P}_{011}} \right\},
        P_{000} \ge \max\nolimits \left\{P_{100},P_{101},P_{110},P_{111}\right\}
    \right)
\end{aligned}\\
&\le P\left(
        {{\mathsf L}_{k-2}}\ge 0, {{\mathsf L}_{k-1}}\ge 0,{{\mathsf L}_{k}}<0, P_{000} \ge \max\nolimits^3 \left\{P_{100},P_{101},P_{110},P_{111}\right\}
    \right)\\
&\approx P\left(
    {{\mathsf L}_{k-2}}\ge 0, {{\mathsf L}_{k-1}}\ge 0,{{\mathsf L}_{k}}<0,
    {{P}_{000}}\ge \frac{P_1}{4}
    \right) \\
&=P\left(
    {{\mathsf L}_{k-2}}\ge 0,{{\mathsf L}_{k-1}} \ge 0,
    -\ln \left( \frac{4e^{{\mathsf L}_{k-2}}}{e^{-{\mathsf L}_{k-1}} + 1}-1 \right)\le{{\mathsf L}_{k}}<0
    \right).
\end{aligned}
\end{equation}

For $P\left( {{\mathsf L}_{k-2}}\ge 0, {{\mathsf L}_{k-1}}<0, {\mathcal C}_k \right)$ illustrated in Fig. \ref{FigCodeTreeSCL4}(c), by Conjecture \ref{ConjectureSCL4}, the probability is bounded as
\begin{equation}\label{EqSCL4Part3Transform}
\begin{aligned}
&P\left( {{\mathsf L}_{k-2}}\ge 0, {{\mathsf L}_{k-1}}<0, {\mathcal C}_k \right) \\
&\begin{aligned}
= &P\left(
        \begin{aligned}
        {{\mathsf L}_{k-2}}\ge 0, {{\mathsf L}_{k-1}}<0, P_{000} \ge P_{001},
        P_{000} \ge \max \left\{ {{P}_{010}},{{P}_{011}} \right\},
        P_{000} \ge \max\nolimits^4 \left\{P_{100},P_{101},P_{110},P_{111}\right\}
        \end{aligned}
    \right) + \\
    &P\left(
        \begin{aligned}
        &{{\mathsf L}_{k-2}}\ge 0, {{\mathsf L}_{k-1}}<0, P_{000} \ge P_{001},
        \min \left\{ {{P}_{010}},{{P}_{011}} \right\}\le P_{000} < \max \left\{ {{P}_{010}},{{P}_{011}} \right\},\\
        &P_{000} \ge \max\nolimits^3 \left\{P_{100},P_{101},P_{110},P_{111}\right\}
        \end{aligned}
    \right) + \\
    &P\left(
        \begin{aligned}
        {{\mathsf L}_{k-2}}\ge 0, {{\mathsf L}_{k-1}}<0, P_{000} \ge P_{001},
        P_{000} < \min \left\{ {{P}_{010}},{{P}_{011}} \right\},
        P_{000} \ge \max\nolimits^2 \left\{P_{100},P_{101},P_{110},P_{111}\right\}
        \end{aligned}
    \right) + \\
    &\underbrace{P\left(
        {{\mathsf L}_{k-2}}\ge 0, {{\mathsf L}_{k-1}}<0, P_{000} < P_{001},
        P_{000} \ge \max \left\{ {{P}_{010}},{{P}_{011}} \right\}
    \right)}_{=0 {\rm~by~Lemma~\ref{LemmaProperty2}}} + \\
    &P\left(
        \begin{aligned}
        &{{\mathsf L}_{k-2}}\ge 0, {{\mathsf L}_{k-1}}<0, P_{000} < P_{001},
        \min \left\{ {{P}_{010}},{{P}_{011}} \right\}\le P_{000} < \max \left\{ {{P}_{010}},{{P}_{011}} \right\},\\
        &P_{000} \ge \max\nolimits^2 \left\{P_{100},P_{101},P_{110},P_{111}\right\}
        \end{aligned}
    \right) + \\
    &P\left(
        \begin{aligned}
        {{\mathsf L}_{k-2}}\ge 0, {{\mathsf L}_{k-1}}<0, P_{000} < P_{001},
        P_{000} < \min \left\{ {{P}_{010}},{{P}_{011}} \right\},
        P_{000} \ge \max \left\{P_{100},P_{101},P_{110},P_{111}\right\}
        \end{aligned}
    \right)
\end{aligned}\\
&\begin{aligned}
\le &P\left(
        \begin{aligned}
        {{\mathsf L}_{k-2}}\ge 0, {{\mathsf L}_{k-1}}<0, P_{000} \ge P_{001},
        P_{000} \ge \max\nolimits^4 \left\{P_{100},P_{101},P_{110},P_{111}\right\}
        \end{aligned}
    \right) + \\
    &P\left(
        \begin{aligned}
        {{\mathsf L}_{k-2}}\ge 0, {{\mathsf L}_{k-1}}<0, P_{000} < P_{001},
        P_{000} \ge \max\nolimits^2 \left\{P_{100},P_{101},P_{110},P_{111}\right\}
        \end{aligned}
    \right)
\end{aligned}\\
&\approx P\left(
    {{\mathsf L}_{k-2}}\ge 0, {{\mathsf L}_{k-1}}<0,
    P_{000} \ge \frac{P_1}{4}
    \right)\\
&\le  P\left(
    {{\mathsf L}_{k-2}}\ge 0, {{\mathsf L}_{k-1}}<0,
    P_{00} \ge \frac{P_1}{4}
    \right) \\
&= P\left(
    \begin{aligned}
    {{\mathsf L}_{k-2}}\ge 0,
    -\ln \left( 4{{e}^{{{\mathsf L}_{k-2}}}}-1 \right)\le{{\mathsf L}_{k-1}}<0
    \end{aligned}
\right).
\end{aligned}
\end{equation}

For $P\left( {{\mathsf L}_{k-2}}< 0, {\mathcal C}_k \right)$ illustrated in Fig. \ref{FigCodeTreeSCL4}(c),
given $\alpha = \frac{\min\left\{P_{100},P_{101},P_{110},P_{111}\right\}}{P_1}$, we have
\begin{equation}\label{EqSCL4Part4Bound}
\begin{aligned}
P\left( {{\mathsf L}_{k-2}}< 0, {\mathcal C}_k \right)
&\le P\left(
    {{\mathsf L}_{k-2}}< 0,
    P_{000} \ge \min \left\{P_{100},P_{101},P_{110},P_{111}\right\}
    \right) \\
&\le P\left(
    {{\mathsf L}_{k-2}}< 0,
    P_{0} \ge \min \left\{P_{100},P_{101},P_{110},P_{111}\right\}
    \right) \\
&= P\left(
    {{\mathsf L}_{k-2}}< 0,
    P_{0} \ge \alpha P_1
    \right) \\
&= P\left( \ln \alpha \le{{\mathsf L}_{k-2}}<0 \right).
\end{aligned}
\end{equation}
Thus, by \eqref{EqSCL4Part1Bound} to \eqref{EqSCL4Part4Bound}, the proof of Proposition \ref{PropositionSCL4PCkUpperBound} is completed.
\end{proof}

\subsection{Approximate Performance of SCL Decoding with $L \ge 8$}

For SCL decoding with $L = 2^m$, $m \ge 3$, we have
\begin{equation}
{\mathcal C}_k = \left\{P_{{\bf 0}_{k-m}^k} \ge \max\nolimits^L\left\{\left.P_{{\bf v}_{k-m}^k}\right| {\bf v}_{k-m}^k \in \left\{0,1\right\}^{m+1}\right\}\right\}
\end{equation}
and $P\left({\mathcal C}_k\right)$ is divided into $(m+2)$ parts, i.e.,
\begin{equation}\label{EqSCL8mADD2Parts}
\begin{aligned}
P\left({\mathcal C}_k\right) =
    &P\left(
        {\mathsf L}_{k-m}\ge 0, \cdots, {\mathsf L}_{k}\ge 0,
        {\mathcal C}_k
        \right) + \\
    &\sum_{i=1}^m P\left(
        {\mathsf L}_{k-m}\ge 0, \cdots, {\mathsf L}_{k-i}\ge 0,
        {{\mathsf L}_{k-i+1}}<0, {\mathcal C}_k
        \right) + \\
    &P\left(
        {{\mathsf L}_{k-m}}< 0, {\mathcal C}_k
        \right).
\end{aligned}
\end{equation}
To calculate \eqref{EqSCL8mADD2Parts}, we provide the superset of ${\mathcal C}_k$ in Lemma \ref{LemmaP0L8largerMINP1paths}.
\begin{lemma}\label{LemmaP0L8largerMINP1paths}
The superset of ${\mathcal C}_k$ is the event that
the path metric of ${\bf 0}_{k-m}^k$ is larger than
the minimum path metric among the attached paths of $v_{k-m} = 1$, i.e.,
\begin{equation}
{\mathcal C}_k \subseteq
\left\{
    P_{{\bf v}_{k-m}^k = {\bf 0}_{k-m}^k} \ge \underset{{\bf v}_{k-m+1}^k \in \left\{0,1\right\}^{m}}{\min} \left\{
        P_{v_{k-m} = 1, {\bf v}_{k-m+1}^k }
    \right\}
\right\}.
\end{equation}
\end{lemma}
\begin{proof}
${\mathcal C}_k$ is equivalent to the event that existing $L$ paths in $\left\{\left.{\bf v}_{k-m}^k\right|{\bf v}_{k-m}^k \in \left\{0,1\right\}^{m+1}, {\bf v}_{k-m}^k \ne {\bf 0}_{k-m}^k\right\}$ with the path metric equal or less than $P_{{\bf v}_{k-m}^k = {\bf 0}_{k-m}^k}$.

Then, existing one path in $\left\{\left.\left(v_{k-m} = 1, {\bf v}_{k-m+1}^k\right)\right|{\bf v}_{k-m+1}^k \in \left\{0,1\right\}^m\right\}$ is included in the $L$ paths. Thus, Lemma \ref{LemmaP0L8largerMINP1paths} is proven.
\end{proof}

Then, we extend Conjecture \ref{ConjectureSCL4} to Conjecture \ref{ConjectureSCL8} as follows.
\begin{conjecture}\label{ConjectureSCL8}
When ${{\mathsf L}_{k-m}}\ge 0$, we assume
\begin{equation}
\underset{{\bf v}_{k-m+1}^k \in \left\{0,1\right\}^{m}}{\min} \left\{
        P_{v_{k-m} = 1, {\bf v}_{k-m+1}^k }
    \right\} \approx \frac{P_1}{L}.
\end{equation}
\end{conjecture}
\begin{remark}
With Conjecture \ref{ConjectureSCL8}, we use $\frac{P_1}{L}$ to replace the minimum path metric among the attached paths of $v_{k-m} = 1$ when ${{\mathsf L}_{k-m}}\ge 0$, which leads to that the derived result is the approximate performance rather than the lower bound.
\end{remark}

With Lemma \ref{LemmaP0L8largerMINP1paths} and Conjecture \ref{ConjectureSCL8}, we provide the upper bound of $P\left({\mathcal C}_k\right)$ in Proposition \ref{PropositionSCL8PCkUpperBound} as follows.
\begin{proposition}\label{PropositionSCL8PCkUpperBound}
The upper bound of $P\left({\mathcal C}_k\right)$ is
\begin{equation}\label{EqPCkSCL8UpperBound}
\begin{aligned}
P\left({\mathcal C}_k\right)  \lesssim
    &P\left(
        {\mathsf L}_{k-m}\ge 0, \cdots, {\mathsf L}_{k}\ge 0
        \right) + \\
    &\sum_{i=1}^m P\left(\begin{aligned}
        {\mathsf L}_{k-m}\ge 0, \cdots, {\mathsf L}_{k-i}\ge 0,
         -\ln \left( \beta_i \right) \le {{\mathsf L}_{k-i+1}}<0 \end{aligned}
        \right) + \\
    &P\left(
        \ln \alpha \le {{\mathsf L}_{k-m}}< 0
        \right),
\end{aligned}
\end{equation}
where
\begin{equation}\label{EqAlphaSCL8}
\alpha = \frac{\underset{{\bf v}_{k-m+1}^k \in \left\{0,1\right\}^{m}}{\min} \left\{
        P_{v_{k-m} = 1, {\bf v}_{k-m+1}^k }
    \right\}}{P_1}
\end{equation}
and
\begin{equation}
\beta_i = \left\{
\begin{aligned}
    &{Le^{L_{k-m}}} - 1 &,{\rm if}~i = m,\\
    &\frac{Le^{L_{k-m}}}{\prod_{q=i}^{m-1}{\left(e^{ - {L_{k-q}}} + 1\right)}} - 1 &,{\rm if}~i\in[\![m-1]\!].
\end{aligned}
\right.
\end{equation}
\end{proposition}
\begin{proof}
For $P\left({\mathsf L}_{k-m}\ge 0, \cdots, {\mathsf L}_{k}\ge 0, {\mathcal C}_k\right)$ and $P\left({{\mathsf L}_{k-m}}< 0, {\mathcal C}_k\right)$,
the proof of Proposition \ref{PropositionSCL8PCkUpperBound} is similar to that of Proposition \ref{PropositionSCL4PCkUpperBound} and we have
\begin{equation}
P\left({\mathsf L}_{k-m}\ge 0, \cdots, {\mathsf L}_{k}\ge 0, {\mathcal C}_k\right) \le P\left({\mathsf L}_{k-m}\ge 0, \cdots, {\mathsf L}_{k}\ge 0\right)
\end{equation}
and
\begin{equation}
P\left({{\mathsf L}_{k-m}}< 0, {\mathcal C}_k\right) \le P\left(\ln \alpha \le {{\mathsf L}_{k-m}}< 0\right).
\end{equation}

Then, according to Lemma \ref{LemmaP0L8largerMINP1paths} and Conjecture \ref{ConjectureSCL8}, we derive the approximate result of the second part in \eqref{EqPCkSCL8UpperBound} as follows.
\begin{equation}\label{EqSCL8Part2Transform}
\begin{aligned}
&P\left({\mathsf L}_{k-m}\ge 0, \cdots, {\mathsf L}_{k-i}\ge 0, {{\mathsf L}_{k-i+1}}<0 , {\mathcal C}_k\right) \\
&\le P\left(
    \begin{aligned}
    {\mathsf L}_{k-m}\ge 0, \cdots, {\mathsf L}_{k-i}\ge 0, {{\mathsf L}_{k-i+1}}<0,
    P_{{\bf v}_{k-m}^k = {\bf 0}_{k-m}^k} \ge \underset{{\bf v}_{k-m+1}^k \in \left\{0,1\right\}^{m}}{\min} \left\{
        P_{v_{k-m} = 1, {\bf v}_{k-m+1}^k }
    \right\}
    \end{aligned}
    \right) \\
&\le P\left(
    \begin{aligned}
    {\mathsf L}_{k-m}\ge 0, \cdots, {\mathsf L}_{k-i}\ge 0, {{\mathsf L}_{k-i+1}}<0,
    P_{{\bf v}_{k-m}^{k-i+1} = {\bf 0}_{k-m}^{k-i+1}} \ge \underset{{\bf v}_{k-m+1}^k \in \left\{0,1\right\}^{m}}{\min} \left\{
        P_{v_{k-m} = 1, {\bf v}_{k-m+1}^k }
    \right\}
    \end{aligned}
    \right) \\
&\approx P\left(
    \begin{aligned}
    {\mathsf L}_{k-m}\ge 0, \cdots, {\mathsf L}_{k-i}\ge 0, {{\mathsf L}_{k-i+1}}<0,
    P_{{\bf v}_{k-m}^{k-i+1} = {\bf 0}_{k-m}^{k-i+1}} \ge \frac{P_1}{L}
    \end{aligned}
    \right) \\
&= P\left(
        \begin{aligned}
        {\mathsf L}_{k-m}\ge 0, \cdots, {\mathsf L}_{k-i}\ge 0,
         -\ln \left(\beta_i\right) \le {{\mathsf L}_{k-i+1}}<0
        \end{aligned}
    \right).
\end{aligned}
\end{equation}
\end{proof}

By Proposition \ref{PropositionSCL2PCkUpperBound}, Proposition \ref{PropositionSCL4PCkUpperBound} and Proposition \ref{PropositionSCL8PCkUpperBound}, we can observe that \eqref{EqPCkSCL8UpperBound} is equivalent to \eqref{EqPCkSCL2UpperBound} and \eqref{EqPCkSCL4UpperBound}.
Thus, the approximate lower bound of SCL decoding is provided in Theorem \ref{TheoremSCLPCkUpperBound} as follows.
\begin{theorem}\label{TheoremSCLPCkUpperBound}
The approximate lower bound of SCL decoding with list size $L = 2^m$ is
\begin{equation}\label{EqTheoremLowerBound}
\begin{aligned}
P\left({\mathcal E}_{\tt SCL}\right)
&= P\left({\mathcal E}_{\tt PL}\right) + P\left({\mathcal E}_{\tt PS}\right) \\
&\gtrsim \sum\limits_{k = m+1}^{K} P\left({\mathcal S}_{k}\right) + P_{\tt ML} \\
&\gtrsim \sum\limits_{k = m+1}^{K}
\left(
\left(1 - P_{\tt UB}\left({\mathcal C}_k\right)\right)
\prod_{j=1}^{k-m-1}{P\left({\mathsf L}_j \ge 0\right)}
\right)
+ P_{\tt ML},
\end{aligned}
\end{equation}
where $P_{\tt UB}\left({\mathcal C}_k\right)$ is the upper bound of $P\left({\mathcal C}_k\right)$ calculated by \eqref{EqPCkSCL8UpperBound}.
\end{theorem}

\begin{remark}
Parameter $\alpha$ affects the calculation result of $P_{\tt UB}\left({\mathcal C}_k\right)$. Generally, due to \eqref{EqAlphaSCL8}, $\alpha$ is assigned with a small number. We also can obtain an upper bound of $P\left(\ln \alpha \le {{\mathsf L}_{k-m}}< 0 \right)$ by discarding $\alpha$, i.e.,
\begin{equation}
P\left(\ln \alpha \le {{\mathsf L}_{k-m}}< 0 \right) \le P\left( {{\mathsf L}_{k-m}}<0 \right).
\end{equation}
\end{remark}

\subsection{Calculation Method Based on GA Algorithm for AWGN Channel}

For AWGN channel, the ML performance in the high SNR region is approximated as the approximate union bound \cite{LinShuBook}, i.e.,
\begin{equation}\label{EqUnionBound}
P_{\tt ML} \approx {{A_{d_{\min}}}Q\left( {\sqrt {\frac{{2{d_{\min}}{E_s}}}{{{N_0}}}} } \right)},
\end{equation}
where $d_{\min}$ is the minimum Hamming weight of the polar code, ${A_{d_{\min}}}$ is the number of the codewords with $d_{\min}$, $E_s$ is the energy of the transmitted signal, $N_0$ is the one-sided power spectral density of AWGN and
\begin{equation}\label{Eq_Q_function}
Q(x)=\frac{1}{{\sqrt {2\pi }  }}\int_x^\infty  {{e^{ - \frac{{{t^2}}}{2}}}dt}
\end{equation}
is the probability that a Gaussian random variable with zero mean and unit variance exceeds the value $x$.

\begin{algorithm}[t]
\setlength{\abovecaptionskip}{0.cm}
\setlength{\belowcaptionskip}{-0.cm}
\caption{$P_{\tt LB}\left({\mathcal E}_{\tt SCL}\right) = {\tt LB\_SCL} \left(N, K, L, {\mathcal A}, \frac{E_s}{N_0} \right)$}\label{AlgorithmSCLLowerBound}
\KwIn {The code length $N$, the information length $K$, the list size $L$, the information set ${\mathcal A}$ and the SNR $\frac{E_s}{N_0}$;}
\KwOut {The lower bound $P_{\tt LB}\left({\mathcal E}_{\tt SCL}\right)$ of SCL decoding;}

Obtain ${\bf v}_1^K$ with $v_k = u_{a_k}$, $k \in [\![K]\!]$, where ${\mathcal A} = \left\{a_1,a_2,\cdots,a_K\right\}$ with $a_1 < a_2 < \cdots < a_K$\;

Calculate $P_{\tt ML}$ by \eqref{EqUnionBound}\;

Calculate $P_{\tt UB}\left({\mathcal C}_k\right)$ by \eqref{EqGAUpperBoundPart1}, \eqref{EqGAUpperBoundPart2} and \eqref{EqGAUpperBoundPart3}\;

Calculate $P_{\tt LB}\left({\mathcal E}_{\tt SCL}\right)$ by \eqref{EqPLBESCL}\;

\Return $P_{\tt LB}\left({\mathcal E}_{\tt SCL}\right)$\;
\end{algorithm}

In GA algorithm, assuming ${\bf v}_1^K = {\bf 0}_1^K$ and ${\bf {\widehat v}}_1^{k-1} = {\bf 0}_1^{k-1}$,
$W_N^{(a_k)}$ is regarded as a BI-AWGN channel and its equivalent noise variance is $\sigma_{a_k}^2$, $k \in [\![K]\!]$. Hence, ${\mathsf L}_k$ obeys Gaussian distribution with the mean $\mu_{{\mathsf L}_k} = \frac{2}{\sigma_{a_k}^2}$ and the variance $\sigma_{{\mathsf L}_k}^2 = \frac{4}{\sigma_{a_k}^2}$
and its probability density function is
\begin{equation}
p\left(l_k\right) = p\left(\left.l_k\right|{\bf {\widehat v}}_1^{k-1} = {\bf 0}_1^{k-1}\right) =
\frac{1}{{\sqrt {2\pi \sigma _{{{\mathsf L}_k}}^2} }}{e^{ - \frac{{{{\left( {{l_k} - {\mu _{{{\mathsf L}_k}}}} \right)}^2}}}{{2\sigma _{{{\mathsf L}_k}}^2}}}}.
\end{equation}

To calculate the lower bound of $P\left({\mathcal E}_{\tt SCL}\right)$, we first calculate the first part of \eqref{EqPCkSCL8UpperBound} as
\begin{equation}\label{EqGAUpperBoundPart1}
\begin{aligned}
P\left({\mathsf L}_{k-m}\ge 0, \cdots, {\mathsf L}_{k}\ge 0\right)
&= \prod_{j=0}^m{P\left(\left.{\mathsf L}_{k-j}\ge 0\right|{\mathsf L}_{k-m}\ge 0,\cdots,{\mathsf L}_{k-j-1}\ge 0\right)} \\
&=\prod_{j=0}^m{P\left({\mathsf L}_{k-j}\ge 0\right)} \\
&=\prod_{j=0}^m{Q\left(-\frac{\mu_{{\mathsf L}_{k-j}}}{ \sigma_{{\mathsf L}_{k-j}}}\right)}.
\end{aligned}
\end{equation}
Then, we calculate the second part of \eqref{EqPCkSCL8UpperBound} as
\begin{equation}\label{EqGAUpperBoundPart2}
\begin{aligned}
&\sum_{i=1}^m P\left({\mathsf L}_{k-m}\ge 0, \cdots, {\mathsf L}_{k-i}\ge 0, -\ln \left( \beta_i \right) \le {{\mathsf L}_{k-i+1}}<0 \right) \\
&=\sum_{i=1}^m \int\limits_{{l}_{k-m}\ge 0, \cdots, {l}_{k-i}\ge 0, -\ln \left( \beta_i \right) \le {{l}_{k-i+1}}<0} {p\left(l_{k-m},\cdots,l_{k-i+1}\right)}\partial {l}_{k-m}\cdots\partial {l}_{k-i}\partial {l}_{k-i+1} \\
&=\sum_{i=1}^m \int\limits_{{l}_{k-m}\ge 0, \cdots, {l}_{k-i}\ge 0} {p\left(l_{k-m},\cdots,{l}_{k-i}\right)}
\left(\int\limits_{-\ln \left( \beta_i \right)} ^ {0}{p\left(\left.{l}_{k-i+1}\right|l_{k-m},\cdots,{l}_{k-i}\right)}
\partial {l}_{k-i+1}\right)
\partial {l}_{k-m}\cdots\partial {l}_{k-i} \\
&=\sum_{i=1}^m \int\limits_{{l}_{k-m}\ge 0, \cdots, {l}_{k-i}\ge 0}
\left(
Q\left(\frac{-\ln \left( \beta_i \right) - \mu_{{\mathsf L}_{k-i+1}}}{ \sigma_{{\mathsf L}_{k-i+1}}}\right) -
Q\left(-\frac{\mu_{{\mathsf L}_{k-i+1}}}{ \sigma_{{\mathsf L}_{k-i+1}}}\right)
\right)
\prod_{j=i}^m{p\left(l_{k-j}\right)}
\partial {l}_{k-m}\cdots\partial {l}_{k-i}.
\end{aligned}
\end{equation}

Finally, the last part of \eqref{EqPCkSCL8UpperBound} is calculated as
\begin{equation}\label{EqGAUpperBoundPart3}
P\left(\ln \alpha \le {{\mathsf L}_{k-m}}< 0\right) =
Q\left(\frac{\ln \alpha - \mu_{{\mathsf L}_{k-m}}}{ \sigma_{{\mathsf L}_{k-m}}}\right) -
Q\left(-\frac{\mu_{{\mathsf L}_{k-m}}}{ \sigma_{{\mathsf L}_{k-m}}}\right).
\end{equation}

By \eqref{EqUnionBound}, \eqref{EqGAUpperBoundPart1},  \eqref{EqGAUpperBoundPart2} and \eqref{EqGAUpperBoundPart3}, the lower bound of $P\left({\mathcal E}_{\tt SCL}\right)$ is determined. The corresponding processes are described in Algorithm \ref{AlgorithmSCLLowerBound}.
In Algorithm \ref{AlgorithmSCLLowerBound}, we first obtain ${\bf v}_1^K$ with $v_k = u_{a_k}$ by ${\mathcal A}$. Then, $P_{\tt ML}$ is calculated by \eqref{EqUnionBound} and $P_{\tt UB}\left({\mathcal C}_k\right)$ is calculated by \eqref{EqGAUpperBoundPart1}, \eqref{EqGAUpperBoundPart2} and \eqref{EqGAUpperBoundPart3}.
Finally, the lower bound $P_{\tt LB}\left({\mathcal E}_{\tt SCL}\right)$ of SCL decoding is determined by
\begin{equation}\label{EqPLBESCL}
\begin{aligned}
P_{\tt LB}\left({\mathcal E}_{\tt SCL}\right)
=
\sum\limits_{k = m+1}^{K}
\left(
\left(1 - P_{\tt UB}\left({\mathcal C}_k\right)\right)
\prod_{j=1}^{k-m-1}{P\left({\mathsf L}_j \ge 0\right)}
\right)
+ P_{\tt ML}.
\end{aligned}
\end{equation}

\section{Construction for SCL Decoding Based on Lower Bound}

\begin{algorithm}[htbp]
\setlength{\abovecaptionskip}{0.cm}
\setlength{\belowcaptionskip}{-0.cm}
\caption{Bit-Swapping Algorithm}\label{AlgorithmConstruction}
\KwIn {The code length $N$, the information length $K$, the list size $L$ and the SNR $\frac{E_s}{N_0}$;}
\KwOut {The information set ${\mathcal A}$;}

$W_N^{(i)}$ is regarded as a BI-AWGN channel and the equivalent noise variance calculated by GA algorithm is $\sigma_i^2$, $i \in [\![N]\!]$\;

Initialize ${\mathcal A}$ by the MWD sequence \cite{arXivMWDsequence} and $P_{\tt LB}\left({\mathcal E}_{\tt SCL}\right) \leftarrow {\tt LB\_SCL} \left(N, K, L, {\mathcal A}, \frac{E_s}{N_0} \right)$\;

Divide the $N$ synthetic channels into $(n+1)$ subsets ${\mathcal B}_r, r = 0,1,\cdots,n, $ by \eqref{EqDividedSubsetBr}\;

Find $c$ making ${\mathcal A} \bigcap {\mathcal B}_c \ne \varnothing$ and ${\mathcal A} \bigcap {\mathcal B}_{c+1} = \varnothing$\;

Set $a \leftarrow c - 1$ and $b \leftarrow c$\;

\While{}
{
    ${\widetilde{\mathcal A}} \leftarrow {\mathcal A}$, ${\mathcal B} \leftarrow \bigcup_{r=a}^b{\mathcal B}_r$,
    $\sigma_{{\max}}^2 \leftarrow \infty$ and $\sigma_{{\min}}^2 \leftarrow 0$\;

    \While{$\sigma_{{\max}}^2 > \sigma_{{\min}}^2$}
    {
        ${\mathcal A}' \leftarrow {\widetilde{\mathcal A}}\bigcap {\mathcal{B}}$ and ${\mathcal F}' \leftarrow {\mathcal B} \setminus {\mathcal A}'$\;

        $i_{\max} \leftarrow {\arg \max_{i \in {\mathcal A}'} }{\left(\sigma_i^2\right)}$ and $\sigma_{{\max}}^2 \leftarrow \sigma_{i_{\max}}^2$\;

        $i_{\min} \leftarrow {\arg \min_{i \in {\mathcal F}'} }{\left(\sigma_i^2\right)}$ and $\sigma_{{\min}}^2 \leftarrow \sigma_{i_{\min}}^2$\;

        ${\widetilde{\mathcal A}} \leftarrow \left(\widetilde{\mathcal{A}}\setminus\left\{i_{\max}\right\} \right)\bigcup \left\{i_{\min}\right\}$\;

        ${\widetilde P}_{\tt LB}\left({\mathcal E}_{\tt SCL}\right) \leftarrow {\tt LB\_SCL} \left(N, K, L, {\widetilde{\mathcal A}}, \frac{E_s}{N_0} \right)$\;

        \If{${\widetilde P}_{\tt LB}\left({\mathcal E}_{\tt SCL}\right) < P_{\tt LB}\left({\mathcal E}_{\tt SCL}\right)$}
        {
            ${\mathcal A} \leftarrow {\widetilde{\mathcal A}}$ and $P_{\tt LB}\left({\mathcal E}_{\tt SCL}\right) \leftarrow {\widetilde P}_{\tt LB}\left({\mathcal E}_{\tt SCL}\right)$\;
        }
    }
    \If{$a = 0$ and $b = n$}
    {
        {\bf break}\;
    }
    \Else
    {
        $a \leftarrow \max{(a - 1, 0)}$ and $b \leftarrow \min{(b + 1, n)}$\;
    }
}
\Return ${\mathcal A}$\;
\end{algorithm}

In this section, we propose a BS algorithm for SCL decoding based on the derived lower bound.
The construction criterion of the construction method is minimizing the lower bound $P_{\tt LB}\left({\mathcal E}_{\tt SCL}\right)$, i.e.,
\begin{equation}\label{EqConstructionCriterion}
{\mathcal A} = \underset{{\mathcal A} \subset [\![N]\!]}{{\arg \min }}{P_{\tt LB}\left({\mathcal E}_{\tt SCL}\right)}.
\end{equation}
To find the optimum ${\mathcal A}$ with the least lower bound of SCL decoding, \eqref{EqConstructionCriterion} searches $C_N^K = \frac{N!}{K!\left(N-K\right)!}$ subsets in $[\![N]\!]$ with the cardinality $K$ and the search complexity is too high.
Inspired by the bit grouping reorder based MWD (BGR-MWD) construction algorithm in \cite{arXivMWDsequence}, we propose a BS algorithm described in Algorithm \ref{AlgorithmConstruction} to greedily search the information set in \eqref{EqConstructionCriterion}.

In Algorithm \ref{AlgorithmConstruction}, we first obtain an initialized ${\mathcal A}$ by the MWD sequence, which has the optimum $P\left({\mathcal E}_{\tt ML}\right)$ but the worse $P\left({\mathcal E}_{\tt PL}\right)$.
Then, to reduce $P\left({\mathcal E}_{\tt PL}\right)$ while remaining $P\left({\mathcal E}_{\tt ML}\right)$ unchanged as possible,
we define a swapping set ${\mathcal B}$ by the Hamming weight and swap the information bit and the frozen bit in ${\mathcal B}$ to gradually reduce $P_{\tt LB}\left({\mathcal E}_{\tt SCL}\right)$ to obtain the optimum ${\mathcal A}$.

Specifically, we first use the MWD sequence to obtain ${\mathcal A}$ in step 2.
Then, similar to the BGR-MWD algorithm, we divide the $N$ synthetic channels into $(n+1)$ subsets in step 3, i.e.,
\begin{equation}\label{EqDividedSubsetBr}
{\mathcal B}_r = \left\{\left.i\right|w({\bf g}_i) = 2^{n-r}, i \in [\![N]\!] \right\}, r = 0,1,\cdots,n,
\end{equation}
where $w({\bf g}_i)$ is the Hamming weight of the $i$-th row of the generator matrix $\bf G$.
Next, we initialize the swapping range $a$ and $b$ in steps 4 and 5, which is used to initialize the swapping set ${\mathcal B} \leftarrow \bigcup_{r=a}^b{\mathcal B}_r$.
In steps 6 to 19, we initialize ${\widetilde{\mathcal A}} \leftarrow {\mathcal A}$ and swap the information bit in ${\mathcal A}' \leftarrow {\widetilde{\mathcal{A}}}\bigcap {\mathcal{B}}$ and the frozen bit in ${\mathcal F}' \leftarrow {\mathcal B} \setminus {\mathcal A}'$ to gradually reduce $P_{\tt LB}\left({\mathcal E}_{\tt SCL}\right)$ in order to optimize ${\mathcal A}$. In steps 9 to 12, we swap the information bit $u_{i_{\max}}$ and the frozen bit $u_{i_{\min}}$ and update
${\widetilde{\mathcal A}} \leftarrow \left(\widetilde{\mathcal{A}}\setminus\left\{i_{\max}\right\} \right)\bigcup \left\{i_{\min}\right\}$, where $i_{\max} = {\arg \max_{i \in {\mathcal A}'} }{\left(\sigma_i^2\right)}$ and $i_{\min} = {\arg \min_{i \in {\mathcal F}'} }{\left(\sigma_i^2\right)}$.
Then, if $\widetilde{\mathcal{A}}$ has the less lower bound than ${\mathcal{A}}$, i.e., ${\widetilde P}_{\tt LB}\left({\mathcal E}_{\tt SCL}\right) < P_{\tt LB}\left({\mathcal E}_{\tt SCL}\right)$, update ${\mathcal A} \leftarrow {\widetilde{\mathcal A}}$ in steps 14 and 15.
If no information bit in ${\mathcal A}'$ has larger equivalent noise variance than the frozen bit in ${\mathcal F}'$, we enlarge the swapping range by
$a \leftarrow \max{(a - 1, 0)}$ and $b \leftarrow \min{(b + 1, n)}$ and update the swapping set ${\mathcal B} \leftarrow \bigcup_{r=a}^b{\mathcal B}_r$. Finally, when the swapping range can not be enlarged, i.e., $a = 0$ and $b = n$, we obtain the optimized ${\mathcal A}$.

\section{Performance Evaluation}

In this section, we first provide the modified upper bound of SC decoding. Then, the lower bound of SCL decoding is illustrated. Finally, the performance of polar codes with the BS algorithm is provided.

\subsection{Simulation Results of Modified Upper Bound of SC Decoding}

\begin{figure*}[t]
\setlength{\abovecaptionskip}{0.cm}
\setlength{\belowcaptionskip}{-0.cm}
  \centering{\includegraphics[scale=0.6]{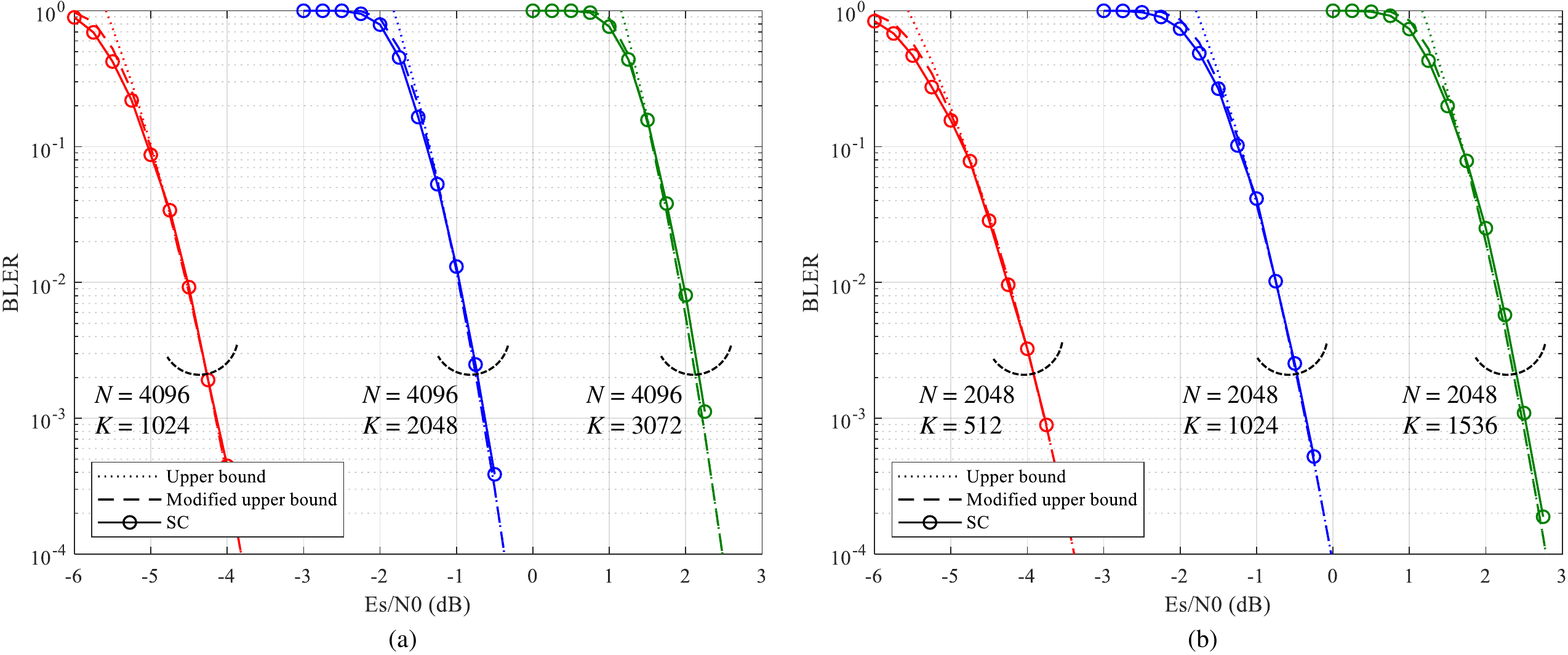}}
  \caption{The comparison between the modified upper bound and the upper bound of SC decoding with $R = 1/4$, $1/2$ and $3/4$. In Fig. \ref{FigBLERSC}(a) and Fig. \ref{FigBLERSC}(b), the improved GA algorithm is used to construct the polar codes with $N = 4096$ and $N = 2048$, respectively.}\label{FigBLERSC}
  \vspace{-0em}
\end{figure*}

Fig. \ref{FigBLERSC} shows the comparison between the modified upper bound \eqref{EqModifiedUpperBoundSC} and the upper bound \eqref{EqSCUpperBound} of SC decoding with $R = 1/4$, $1/2$ and $3/4$.
In Fig. \ref{FigBLERSC}(a) and Fig. \ref{FigBLERSC}(b), the improved GA algorithm \cite{GA_DAI} is used to construct polar codes with $N = 4096$ and $N = 2048$, respectively. In Fig. \ref{FigBLERSC}, we can observe that the proposed modified upper bound is tighter than the upper bound in the low SNR region and coincides with the upper bound in the medium to high SNR region.
The reason is that compared with the upper bound \eqref{EqSCUpperBound}, the modified upper bound in \eqref{EqModifiedUpperBoundSC} considers the influence of the prior information bits and the calculation process introduces the path metric of the prior information bits.
Moreover, since the path metric of the prior information bits is almost equal to 1 for the large SNR, the two bounds coincide with each other in the high SNR region.
Then, we observe that the modified upper bound is almost identical to the BLER performance of SC decoding in low to high SNR region. Thus, we can use the modified upper bound to evaluate the BLER performance of SC decoding without simulation.

\subsection{Simulation Results of Lower Bound of SCL Decoding}

\begin{figure*}[t]
\setlength{\abovecaptionskip}{0.cm}
\setlength{\belowcaptionskip}{-0.cm}
  \centering{\includegraphics[scale=0.6]{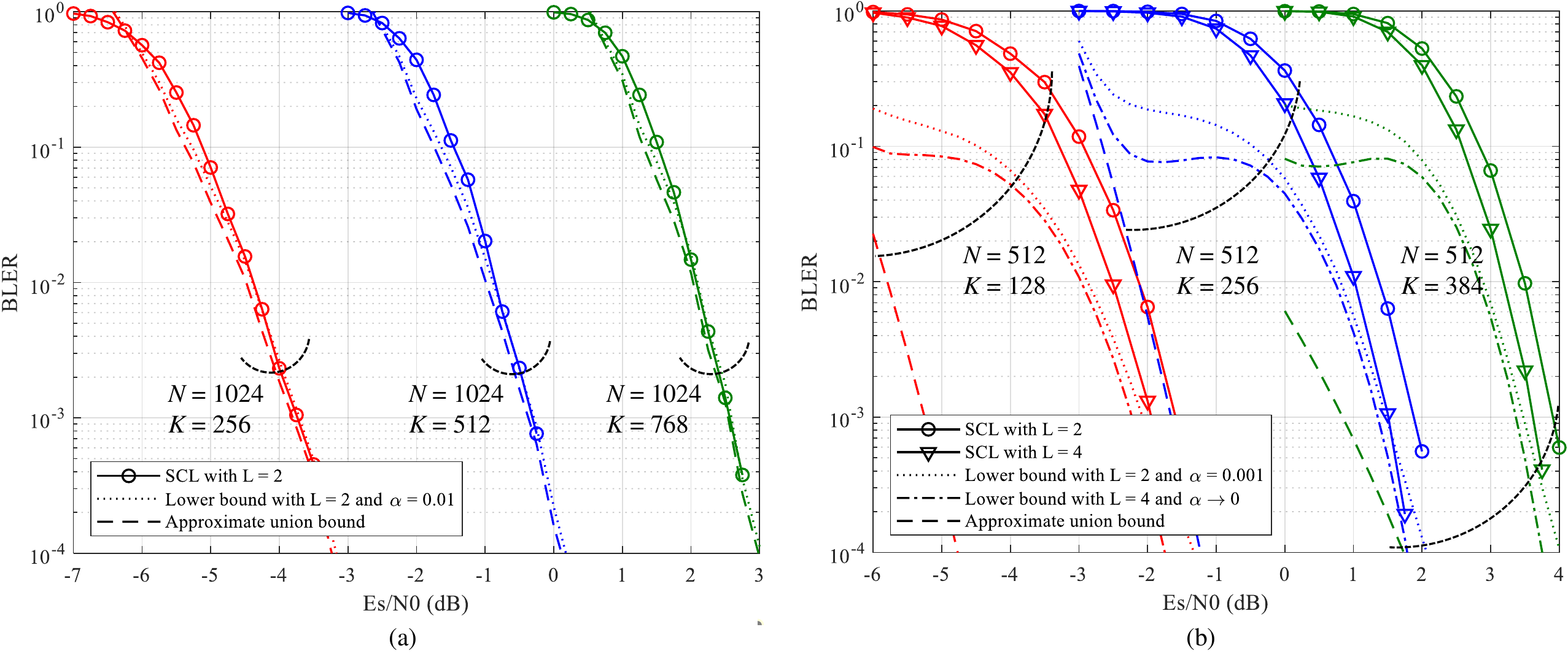}}
  \caption{The comparison among the lower bound of SCL decoding, the BLER performance of SCL decoding and the approximate union bound  of polar codes
  with $R = 1/4$, $1/2$ and $3/4$, where the polar codes with $N = 1024$ are constructed by the improved GA algorithm in Fig. \ref{FigBLERSCL2}(a) and the MWD sequence is used to construct the polar codes with $N = 512$ in Fig. \ref{FigBLERSCL2}(b).}\label{FigBLERSCL2}
  \vspace{-0em}
\end{figure*}

Fig. \ref{FigBLERSCL2} illustrates the comparison among the lower bound \eqref{EqTheoremLowerBound} of SCL decoding, the BLER performance of SCL decoding and the approximate union bound \eqref{EqUnionBound} of polar codes with $R = 1/4$, $1/2$ and $3/4$. In Fig. \ref{FigBLERSCL2}(a), the polar codes with $N = 1024$ are constructed by the improved GA algorithm \cite{GA_DAI}. In Fig. \ref{FigBLERSCL2}(b), the MWD sequence \cite{arXivMWDsequence} is used to construct the polar codes with $N = 512$.
We can observe in Fig. \ref{FigBLERSCL2}(a) that the BLER performance, the lower bound and the approximate union bound are close to each other and the curve of the lower bound is located between the other two curves.
Then, in the high SNR region, the lower bound coincides with the approximate union bound and the BLER performance, which means that the probability of the PL error event is less than the ML performance and the ML error event is dominant in the error event of SCL decoding.
With this, we provide an explanation for the phenomenon that the BLER performance of polar codes constructed by GA algorithm and decoded by SCL decoding with limited list size can approach the ML performance.
In Fig. \ref{FigBLERSCL2}(b), we observe that the BLER performance of SCL decoding gradually approaches the corresponding lower bound as the SNR increases.
Specifically, for $\left(512, 256\right)$ polar code at BLER $10^{-3}$, there is about $0.15$dB performance gap between the performance of SCL decoding with $L = 4$ and the corresponding lower bound.
Hence, the proposed lower bound of SCL decoding can be used to evaluate the BLER performance in the medium to high SNR region.
Then, we also observe that the performance gap between the lower bound and the approximate union bound is large, which means the PL error event is dominant in the error event of polar codes constructed by the MWD sequence under SCL decoding rather than the ML error event, though the MWD sequence has been proved to be the optimum construction method under ML decoding in \cite{arXivMWDsequence}.

\begin{figure*}[t]
\setlength{\abovecaptionskip}{0.cm}
\setlength{\belowcaptionskip}{-0.cm}
  \centering{\includegraphics[scale=0.6]{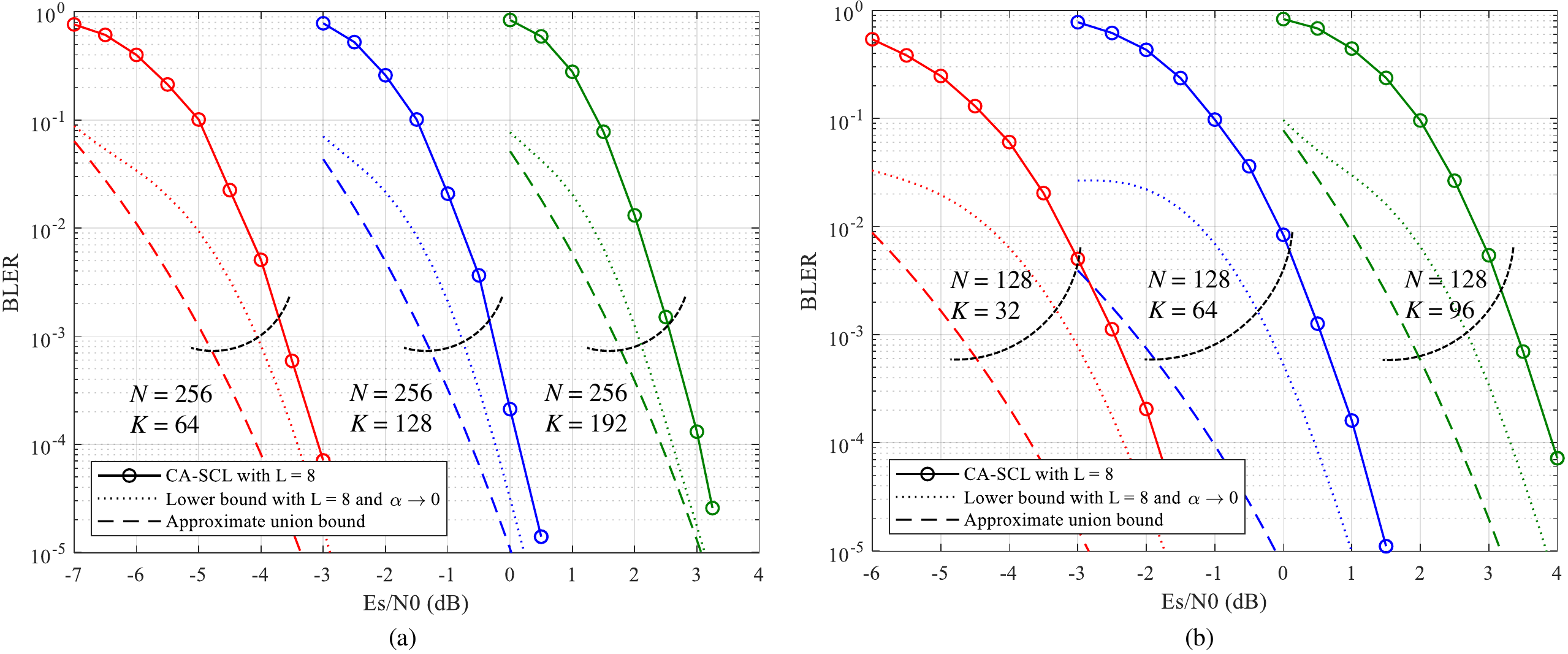}}
  \caption{The comparison among the lower bound of CA-SCL decoding, the BLER performance of CA-SCL decoding and the approximate union bound of CRC-polar concatenated codes with $R = 1/4$, $1/2$ and $3/4$, where the construction method is PW and the CRC lengths in Fig. \ref{FigBLERSCL8CRC}(a) and Fig. \ref{FigBLERSCL8CRC}(b) are 8 and 11, respectively.}\label{FigBLERSCL8CRC}
  \vspace{-0em}
\end{figure*}

Fig. \ref{FigBLERSCL8CRC} shows the comparison among the lower bound of CRC-aided SCL (CA-SCL) decoding, the BLER performance of CA-SCL decoding \cite{CASCL} and the approximate union bound of CRC-polar concatenated codes with $R = 1/4$, $1/2$ and $3/4$, where the construction method is PW.
In Fig. \ref{FigBLERSCL8CRC}(a), the code length is $N = 256$ and the CRC is the optimum 8-bit CRC in \cite{SphereMWD}. In Fig. \ref{FigBLERSCL8CRC}(b), we use the optimum 11-bit CRC in \cite{SphereMWD} for the CRC-polar concatenated codes with $N = 128$.
In Fig. \ref{FigBLERSCL8CRC}, we observe that the BLER performance approaches the lower bound in the medium to high SNR region and there is a large performance gap between the BLER performance and the approximate union bound.
Specifically, for $\left(128, 64\right)$ CRC-polar concatenated code with the optimum 11-bit CRC at BLER $10^{-4}$, the gap between the BLER performance and the lower bound is about $0.61$dB, which is less than the gap about $2.11$dB between the BLER performance and the approximate union bound.
Thus, the PL error event is dominant in the error event of CRC-polar concatenated codes with optimum CRC and
the proposed lower bound can be used to evaluate the BLER performance of CRC-polar concatenated codes under CA-SCL decoding.

\subsection{Simulation Results of Bit-Swapping Algorithm}

\begin{figure}[t]
\setlength{\abovecaptionskip}{0.cm}
\setlength{\belowcaptionskip}{-0.cm}
  \centering{\includegraphics[scale=0.6]{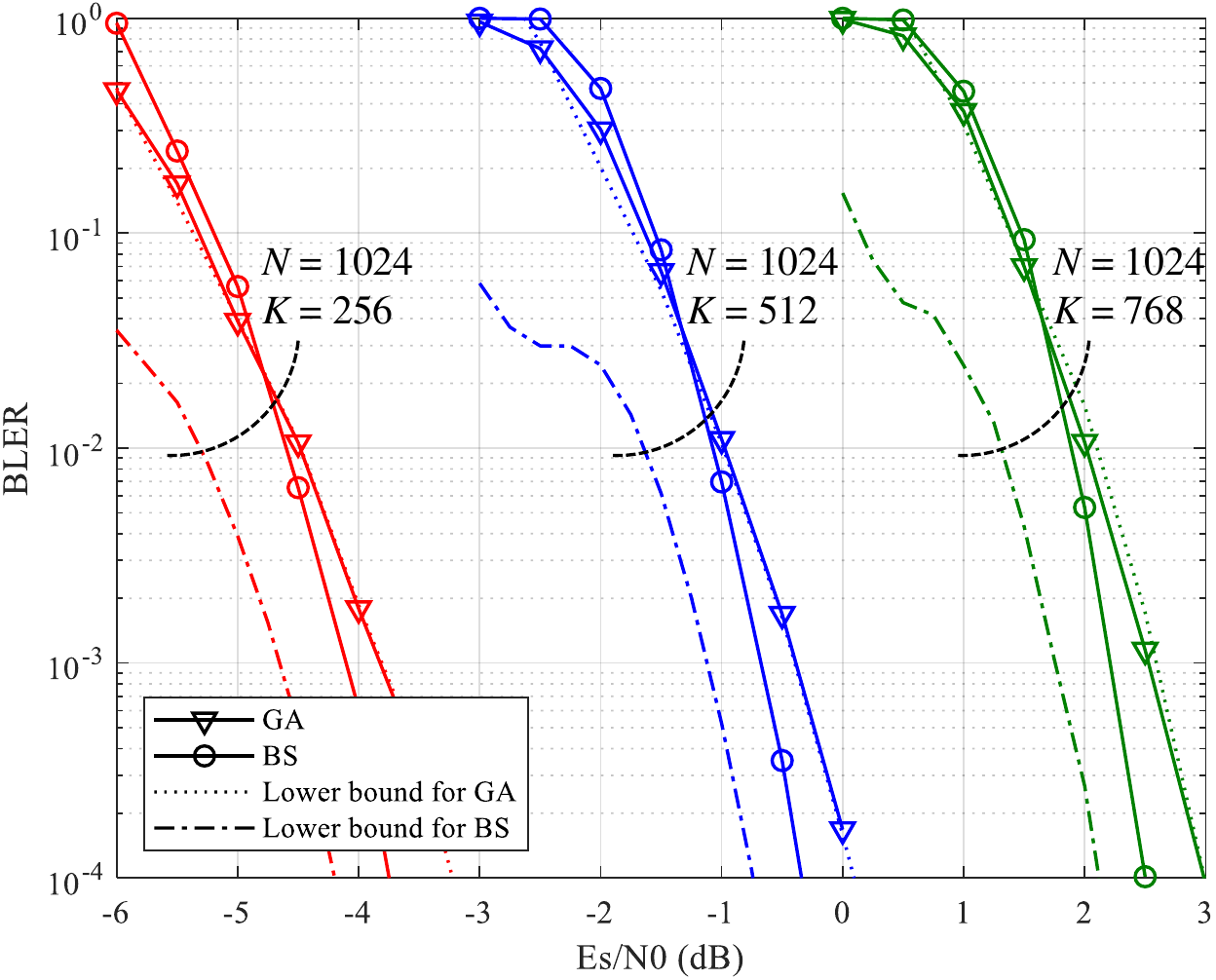}}
  \caption{The BLER performance of polar codes constructed by the proposed BS algorithm under SCL decoding with $L = 4$.}\label{FigN1024SCL4BS}
  \vspace{-0em}
\end{figure}

Fig. \ref{FigN1024SCL4BS} provides the BLER performance of polar codes constructed by the proposed BS algorithm under SCL decoding with $L = 4$.
In Fig. \ref{FigN1024SCL4BS}, we observe that compared with the GA algorithm, the polar codes constructed by the BS algorithm have about $0.32$dB, $0.29$ and $0.32$dB performance gain with $R = 1/4$, $R = 1/2$ and $R = 3/4$ at BLER $10^{-3}$, respectively.
Then, as the SNR increases, the BLER performance with BS algorithm gradually approaches the corresponding lower bound.
Thus, the proposed lower bound can be treated as a construction metric to optimize polar codes under SCL decoding.

\section{Conclusion}

In this paper, we analyze the BLER performance of polar codes under SCL decoding and propose the corresponding lower bound. We indicate that the error event of SCL decoding consists of the PL error event and the PS error event. With this, we use the ML performance to approximate the probability of the PS error event and
derive the lower bound of the probability of the PL error event to provide the lower bound of SCL decoding.
Based on the lower bound, we design the BS algorithm to construct polar codes under SCL decoding by reducing the lower bound greedily.
The simulation results show that the BLER performance of SCL decoding approaches the lower bound in the medium to high SNR region and the proposed BS algorithm can improve the BLER performance of polar codes under SCL decoding.

\bibliographystyle{IEEEtran}
\bibliography{IEEEabrv,myrefs}

\end{document}